\documentclass{article}
\usepackage[a4paper, left=1.2in,right=1.2in,top=1.5in,bottom=1.5in]{geometry}
\usepackage[hyphens]{url}  
\usepackage{hyperref}       
\usepackage{booktabs}       
\usepackage{amsfonts}       
\usepackage{nicefrac}       
\usepackage[sort, numbers]{natbib}
\usepackage{amsfonts,amsthm}
\usepackage{bm}
\usepackage{algorithm}
\usepackage{algpseudocode}
\usepackage{multirow}
\usepackage{xcolor}
\usepackage{tikz}
\usepackage{mathtools}
\usepackage{balance}
\usepackage{authblk}
\usepackage{xspace}

\usetikzlibrary{positioning,chains,fit,shapes,calc}

\definecolor{myblue}{RGB}{80,80,160}
\definecolor{mygreen}{RGB}{80,160,80}
\providecommand{\NumAgent}{\ensuremath{n}\xspace} 
\providecommand{\Util}[2][]{\ensuremath{ 
\ifthenelse{\equal{#1}{}}{u_{#2}}{u_{#2}^{#1}}}\xspace}  

\providecommand{\aPSNE}{\ensuremath{\text{$\epsilon$-PSNE}}\xspace} 
\providecommand{\G}{\ensuremath{\mathcal{G}}\xspace}    
\providecommand{\V}{\ensuremath{\mathcal{V}}\xspace}    
\providecommand{\E}{\ensuremath{\mathcal{E}}\xspace}    
\providecommand{\Cost}[1]{\ensuremath{c_{#1}}\xspace}    
\providecommand{\g}[1]{\ensuremath{g_{#1}}\xspace}    
\providecommand{\InvNum}[2]{\ensuremath{ n_{#1}^{(#2)} }\xspace}    
\providecommand{\N}[1]{\ensuremath{\mathcal{N}_{#1}}\xspace}    
\providecommand{\Action}[1]{\ensuremath{x_{#1}}\xspace}    
\providecommand{\Dg}[2]{\ensuremath{ \Delta g_{#1}(#2) }\xspace}    
\providecommand{\Dgh}[1]{\ensuremath{ \Delta g(#1) }\xspace}    
\providecommand{\ActionVec}{\ensuremath{\bm{x}}\xspace} 
\providecommand{\BNPG}[1]{\ensuremath{\textsc{BNPG}(#1)}\xspace}
\providecommand{\I}[1]{\ensuremath{\mathcal{I}_{#1}}\xspace}
\providecommand{\T}{\ensuremath{\mathcal{T} }\xspace}
\providecommand{\dMax}{\ensuremath{ d_{\text{max}} }\xspace} 
\providecommand{\numEvo}{\ensuremath{k}\xspace}
\providecommand{\StopCond}{\ensuremath{\delta}\xspace}

\providecommand{\Norm}[2][]{\ensuremath{%
\ifthenelse{\equal{#1}{}}{\|{#2}\|}{\|{#2}\|_{{#1}}}}\xspace}
\providecommand{\SetCard}[1]{\ensuremath{| #1 |}\xspace}
\providecommand{\SET}[1]{\ensuremath{\{ #1 \}}\xspace}

\providecommand{\Set}[2]{\ensuremath{\SET{#1 \mid #2}}\xspace}

\newtheorem{defn}{Definition}[section]
\newtheorem{theorem}{Theorem}
\newtheorem{proposition}[theorem]{Proposition}
\newtheorem{corollary}[theorem]{Corollary}

\allowdisplaybreaks

\usepackage{enumitem}
\definecolor{myblue}{RGB}{80,80,160}
\definecolor{mygreen}{RGB}{80,160,80}

\title{Computing Equilibria in Binary Networked Public Goods Games}
\author[1]{Sixie Yu \thanks{The first two authors contributed equally to the paper.} \thanks{sixie.yu@wustl.edu}}
\author[2]{Kai Zhou}
\author[3]{P. Jeffrey Brantingham}
\author[1]{Yevgeniy Vorobeychik}
\affil[1]{Washington University in St. Louis}
\affil[2]{The Hong Kong Polytechnic University}
\affil[3]{University of California, Los Angeles}

\numberwithin{theorem}{section}
\begin{document}
\date{ }
\maketitle

\begin{abstract}
Public goods games study the incentives of individuals to contribute
 to a public good and their behaviors in equilibria. In this paper, we
 examine a specific type of public goods game where players are
 networked and each has binary actions, and focus on the algorithmic
 aspects of such games. First, we show that checking the
 existence of a pure-strategy Nash equilibrium is NP-complete. 
 We then identify tractable instances based on restrictions of either utility
 functions or of the underlying graphical structure.  In certain
 cases, we also show that we can efficiently compute a socially optimal Nash
 equilibrium.
Finally, we propose a heuristic approach for computing approximate equilibria
in general binary networked public goods games, and experimentally
demonstrate its effectiveness.
\end{abstract}

\section{Introduction}

Public goods games have been a major subject of inquiry as a natural way to model the tension between individual interest and social good~\cite{kollock1998social,santos2008social}.
In a canonical version of such games, individuals choose the amount to invest in a public good, and the subsequent value of the good, which is determined by total investment of the community, is then shared equally by all.
A number of versions of this game have been studied, including variants where players are situated on a network, with individual payoffs determined solely by the actions of their network neighbors~\cite{bramoulle2007public}.

An important variation of networked public goods games treats investment decisions as \emph{binary}~\cite{galeotti2010network}.
One motivating example is crime reporting, known to vary widely, with the general observation that crime is considerably under-reported~\cite{Morgan17}.
In a game theoretic abstraction of crime reporting, individuals choose whether or not to report crimes, and benefits accrue to the broader community, for example, causing a reduction in overall crime rate.
Another example of binary public goods games is vaccination.  In this example, parents decide whether to vaccinate their children, with herd immunity becoming the public good.
To keep terminology general, we will refer to the binary decision whether or not to \emph{invest} in a public good (thus, reporting a crime and vaccinating are forms of such investment).

A special case of binary public goods games (BNPGs), commonly known as \emph{best-shot games}, has received some prior attention~\cite{Dallasta11,galeotti2010network,Komarosvky15,Levit18}.
However, best-shot games make a strong assumption about the structure of the player utility functions.
While \citet{Levit18} recently showed that equilibria for best-shot games can be computed in polynomial time by best response dynamics, these may fail to find social welfare-optimal equilibria, and nothing is known about BNPGs more generally.

We study the problem of computing pure strategy Nash equilibria in general binary public goods games on networks.
We provide examples to show that a pure strategy Nash equilibria is not guaranteed to exist for general BNPGs.
Furthermore, we show that even determining if a pure strategy Nash equilibrium exists is hard in general. 
However, while pure strategy equilibria need not in general exist, and are hard to compute, we exhibit a number of positive results for important special cases.
One class of special cases pertains to binary public goods games for completely connected networks (communities), in which case an equilibrium can be found efficiently if it exists.
Similarly, we can efficiently find a pure strategy equilibrium in general BNPGs when the graph is a tree.
If we further restrict the externalities to have an identical impact on all players (we call this the \emph{homogeneous} case), we can characterize \emph{all} pure strategy equilibria, and efficiently compute a socially optimal equilibrium.
Moreover, if both investment costs and externalities are identical (termed the \emph{fully homogeneous} case), we can characterize \emph{all} pure strategy equilibria for arbitrary networks.
Finally, we present a heuristic approach for finding approximate equilibria to tackle general BNPGs, and experimentally demonstrate that it is highly effective.
Our algorithmic results are summarized in Table~\ref{tab:complexity}.

\subsection{Related Work} 
Our work relates to the broad literature on graphical games, a succinct representation of games in which utilities exhibit local dependences captured by a graph~\cite{kearns2013graphical}.
 \cite{elkind2006nash,daskalakis2009complexity,gottlob2005pure} studied the complexity of computing (mixed-strategy) equilibria in graphical games with utilities defined in the matrix form, and proposed efficient algorithms to find equilibria on graphs with restricted topologies.

Networked public goods games can be regarded as a special class of graphical games where the utilities are functions of the accumulated efforts of individuals. A number of model variations have been proposed to study public goods in different fields such as economics, innovation diffusion, and medical research \cite{feick1987market,valente1996network,burt1987social}. 
Our model is closely related to that proposed by ~\citet{bramoulle2007public}. 
The two qualitative distinctions are that (a) we consider binary investment decisions, in contrast to \citet{bramoulle2007public}, who focus on the more traditional continuous investment model, and (b) \citet{bramoulle2007public} assume homogeneous concave utilities, while we consider a more general setting.

Other related variations of graphical games include supermodular network games \cite{manshadi2009supermodular} and best-shot games \cite{Dallasta11,galeotti2010network,Komarosvky15,Levit18}. 
The latter are a special case of BNPGs, and \citet{Levit18} recently showed that these are potential games with better response dynamics converging to a pure strategy equilibrium in polynomial time.
No other algorithmic results are known for either of these special classes of graphical games.


\begin{table}[h]
\centering
\small
\begin{tabular}{@{}cccc@{}}
\toprule
                          & general                  & complete graph      & tree      \\ \midrule
heterogeneous             & hard                       & poly                & poly                            \\
homogeneous               & hard                       & poly                & poly                           \\
fully-homogeneous         & hard                       & poly                & poly                            \\
fully-homogeneous + convex $g$ & poly                  & poly                & poly                           \\
\bottomrule
\end{tabular}
\caption{An overview of our results.}
\label{tab:complexity}
\end{table}

\section{Model}

A Binary Networked Public Goods game (henceforth BNPG game) is characterized by a graph $\mathcal{G} = (\V, \E)$; throughout this chapter we assume that \G is simple, undirected, and loop-free.
The players consist of the node set $\V = \SET{1,2,\ldots,n}$; the edge set  $\E = \Set{(i,j)}{i,j \in \V}$ represents the interdependencies among the players.
Let $\Action{i} \in \SET{0, 1}$ be player $i$'s action; we can think of $\Action{i}=1$ (resp. $\Action{i}=0$) as the player decides to invest (resp. not invest) in some public goods.
The action profile of all players is represented by $\ActionVec \in \SET{0,1}^n$ and the action profile of all players other than player $i$ is $\ActionVec_{-i}$.

Let $\N{i} = \Set{j \in \V}{(i, j) \in \E}$ be the neighbors of $i$.
Define $\InvNum{i}{\ActionVec} = \sum_{j\in \N{i}} x_j$ as the number of $i$'s  neighbors who choose to invest, e.g., in the crime reporting scenario $\InvNum{i}{\ActionVec}$ represents the number of $i$'s neighbors who report crimes.
Reporting crime is costly, e.g., time consuming or even dangerous.
We capture the cost of each player with a constant $\Cost{i}$.
Player $i$'s utility function is defined below
    \begin{equation}
        \label{eqn-cost}
            \Util{i}(\ActionVec) = \Util{i}(\Action{i}, \ActionVec_{-i}) = \Util{i}(\Action{i}, \InvNum{i}{\ActionVec}) = \g{i}(\Action{i} + \InvNum{i}{\ActionVec}) - \Cost{i} \Action{i},
    \end{equation}
where $\g{i}(\Action{i} + \InvNum{i}{\ActionVec})$ captures the positive externality that $i$ experiences from her neighbors' (and her own) investment.
As is standard in public goods literature, we assume that \g{i} is a \emph{non-decreasing} function; in the case of crime reporting this means that higher reporting reduces crime rate.
Our definition of $\g{i}$ generalizes the one given by \citet{bramoulle2007public}; they assume that $\g{i}$ is a twice differentiable strictly concave function, which may be violated in some scenarios. 
For example, studies of incentives in P2P systems entail natural models in which $\g{i}$ is an S-shaped function~\cite{buragohain2003game}; commonly considered best-shot games model $\g{i}$ as a step function~\cite{galeotti2010network}; $\g{i}$ may be convex as in models of social unrest~\cite{chwe2000communication}.
Our definition of $\g{i}$, on the other hand,  is flexible enough to model the aforementioned scenarios.
Notice that each function $\g{i}$ can be represented using $O(n)$ values, so the entire BNPG game (including the graph structure) can be represented using $O(n^2)$ values.
A useful observation is that player $i$ chooses to invest only if  $\Util{i}(1, \InvNum{i}{\ActionVec}) \ge \Util{i}(0, \InvNum{i}{\ActionVec})$,
which is rewritten as 
    \begin{equation}\label{eqn-condition-2}
            \Dg{i}{\InvNum{i}{\ActionVec}} := \g{i}(\InvNum{i}{\ActionVec} + 1) - \g{i}(\InvNum{i}{\ActionVec}) \ge \Cost{i}.
    \end{equation}
The finite difference $\Dg{i}{\InvNum{i}{\ActionVec}}$ will be useful in later analysis.
In general, a BNPG is called \emph{heterogeneous} if the players have different utility functions, i.e., different \Cost{i} and \g{i}.
As we discuss later, checking the existence of a Nash equilibrium in a heterogeneous BNPG is hard.
Due to the hardness result, we also study several special cases such that positive algorithmic results can be obtained; intuitively, we make different assumptions about the utility functions or the underlying graph.
For example, the two cases below restrict the externality functions or the cost parameters of the players.
    \begin{description}
        \item[Homogeneous BNPG:] The players share the same externality function, i.e., $\g{1} = ,\ldots, = \g{n}$.
        \item[Fully Homogeneous BNPG:] In addition to having the same externality function, the players share the same cost, i.e., $\Cost{1} = ,\ldots, = \Cost{n}$.
    \end{description}

We use $\BNPG{\G, \mathcal{U} }$ to represent an instance of BNPG with  underlying graph $\G$; the set $\mathcal{U}$ consists of the players' utility functions, i.e., $\mathcal{U}=\Set{\Util{i}}{ i \in \V}$, which contains all the information about \g{i} and \Cost{i}.
Given a BNPG instance, We seek \emph{pure-strategy} Nash equilibria of the game. 
A pure-strategy Nash equilibria (PSNE) is defined as follows:
\begin{defn}
	Given $\BNPG{\G,\mathcal{U}}$, 
	a pure-strategy Nash equilibrium (PSNE) of the game is an action profile $\ActionVec \in \SET{0, 1}^n$ such that for every player $\Util{i}(\Action{i}, \InvNum{i}{\ActionVec}) \ge \Util{i}(1-\Action{i}, \InvNum{i}{\ActionVec})$.
\end{defn}

We differentiate two types of PSNE: 1) \emph{trivial} PSNEs, including $\ActionVec = \bm{1}$ (i.e., all players invest) and $\ActionVec=\bm{0}$ (i.e., no player invests) and 2) \emph{non-trivial} PSNEs, i.e., any action profile $\ActionVec \notin \SET{\bm{0}, \bm{1}}$.
The formal definitions are as follows:
    \begin{defn}
        Given a BNPG, suppose at least one PSNE exists.
        Let $\ActionVec \in \SET{0, 1}^n$ be a PSNE.
        If $\ActionVec \in \SET{\bm{0}, \bm{1}}$, it is a trivial PSNE; otherwise, it is a non-trivial PSNE.
    \end{defn}

Given a BNPG, suppose it has at least one non-trivial PSNE with $0 < k < \NumAgent$ investing players.
We identify two types of players in the PSNE: 1) the players who \emph{always} invest  and 2) the players who \emph{never} invest.
Intuitively, the two types of players either invest or do not invest, no matter what other players' actions are. 
The formal statement is as follows:
    \begin{proposition}\label{P:always}
        Given a BNPG, suppose it has at least one non-trivial PSNE with $0 < k < \NumAgent$ investing players.
        Define the following two sets:
        \begin{equation}
            \begin{aligned}
	            & \I{+}(k):=\Set{i \in \V}{ \Cost{i} < \Dg{i}{k}    }, \\
	            & \I{-}(k):=\Set{i \in \V}{ \Cost{i} > \Dg{i}{k-1}  }.
            \end{aligned}
        \end{equation}
        It follows that the set $\I{+}(k)$ includes the players who always invest while the set $\I{-}(k)$ the players who never invest; in addition, there is no \Cost{i} such that $\Dg{i}{k-1} < \Cost{i} < \Dg{i}{k}$.
    \end{proposition}
    \begin{proof}
    Suppose there is a player $j \in \I{+}(k)$ who may not invest; it follows that $\Util{j}(0, k) \ge \Util{j}(1, k)$ and thus $\Cost{j} \ge \Dg{j}{k}$, which contradicts the definition of $\I{+}(k)$.
    Similarly, the set $\I{-}(k)$ includes the players who never invest.
    The argument is similar: suppose a player $j \in \I{-}(k)$ who may invest; it follows that $\Util{j}(1, k-1) \ge \Util{j}(0, k-1)$ and thus $\Dg{j}{k-1} \ge \Cost{i}$, a contradiction to the definition.
    Since $\I{+}(k)$ is disjoint from $\I{-}(k)$, we have that there is no \Cost{i} such that $\Dg{i}{k-1} < \Cost{i} < \Dg{i}{k}$.
    \end{proof}


\section{Heterogeneous BNPG}

We first construct an example to show that a heterogeneous BNPG does not necessarily have a PSNE.
Moreover,  we show that checking the existence of a PSNE is NP-complete.
Due to the hardness result, we restrict the game in several aspects in order to obtain positive algorithmic results.
In particular, we restrict the underlying graph \G to a complete graph or a tree, which leads to a polynomial-time algorithms to compute a PSNE (or conclude no one exists). 
For general \G, we describe a heuristic algorithm to find an approximate PSNE.

We start by constructing a BNPG where no PSNE exists.
The game has two players $A$ and $B$. 
Let $\Util{A}(x_A, x_B) =  g_A(x_A+x_B) - c_A x_A$ and $\Util{B}(x_B, x_A)= g_B(x_B + x_A) - c_B x_B$ be the utility functions of $A$ and $B$, respectively. 
Consider the setting where  $g_A(0)=\epsilon$, $g_A(1)=c_A$, $g_A(2)=2c_A + \epsilon$, and $g_B(0)=\epsilon$, $g_B(1)=c_B + 2\epsilon$, $g_B(2)=2c_B + \epsilon$, where $\epsilon \le c_A$ and $\epsilon \leq c_B$, ensuring that both $g_A$ and $g_B$ are non-decreasing.
By checking the four possible action profiles of the game, it is direct to verify that no PSNE exists.

Next, we show that checking if a heterogeneous BNPG has a PSNE is NP-complete.

\begin{theorem}\label{thm-hard}
Given a heterogeneous BNPG, checking the existence of a PSNE is NP-complete.
\end{theorem}
\begin{proof}
Given an action profile, it takes polynomial time to check if it is a PSNE, so the problem is in NP.
We construct a reduction from the \textsc{Independent Set} (IS) problem.
Given a graph $\G$, the \textsc{IS} problem is to decide whether there is an independent set $S$ of size at least $k$.
Given an instance of \textsc{IS} (i.e., a graph $\G=(\V', \E)$ and an integer $k$), we construct a new graph $\mathcal{H}=(\V, \E)$ by adding an additional node $t$ and connecting $t$ to every node in \G, that is $\V = \V' \cup \SET{t}$.
Consider the following best-response policy for the players:
\begin{equation}\label{eq:best-response-hete}
	\begin{aligned}
		i \in \V' : \Action{i} = \begin{cases}
							1, & \InvNum{i}{\ActionVec} = 0 \\
							0, & \InvNum{i}{\ActionVec} \ne 0
					   \end{cases} 
		\text{ and }      
		            \Action{t} = \begin{cases}
	    				   1, & 0 < \InvNum{t}{\ActionVec} < k \\
	    				   0, & \InvNum{t}{\ActionVec} = 0 \text{ or } \InvNum{t}{\ActionVec} \ge k.
	    			   \end{cases}
	\end{aligned}
\end{equation}
One way to realize the above is as follows: for any player $i \in \V'$, let $\g{i}(0)=0$,  $\g{i}(x)=1$ for any $x > 0$, and $\Cost{i}=0.5$; for player $t$, let $\g{t}(0) = \g{t}(1) = 0$, $\g{t}(x) = x$ for $1 < x \le k$, $\g{t}(x)=k$ for $x \ge k+1$, and $\Cost{t}=0.5$. 
A heterogeneous BNPG is constructed on the graph $\mathcal{H}$, with the utility functions as $\g{i}$ and $\g{t}$.
We show that finding an independent set of size at least $k$ in graph \G is equivalent to finding a PSNE of the game. 

First, suppose we are given an independent set $S$ of size $k$. 
Then node $t$ is not in $S$ since $t$ connects
to every node in \G. 
We check the remaining nodes in $\V \setminus S$ and iteratively add into $S$ those nodes that have no connections to any node in $S$, which results in $\hat{S}$. Note that $\hat{S}$ is maximal and its size is at at least $k$.
We argue that all nodes in $\hat{S}$ choosing $1$ and
the rest of the nodes choosing $0$ is an equilibrium. For a node
in $\hat{S}$, it will not deviate since none of its neighbors chooses $1$.
For node $t$, it will not deviate since it has at least $k$ neighbors that choose $1$. For a node $u \in \V \setminus \hat{S}$, it must connect to at least one node in $\hat{S}$; otherwise $\hat{S}$ is not maximal. Thus $u$ will also not deviate.

Second, suppose we are given a PSNE \ActionVec. 
There must be some nodes choosing $1$ in the equilibrium; otherwise, for an arbitrary node in \G it has no investing neighbor and it would prefer choosing $1$.
If the equilibrium is that all nodes choose $1$, we must have $\InvNum{t}{\ActionVec} < k$; however, recall that $t$ connects to all nodes in \G, i.e., $\InvNum{t}{\ActionVec} = \SetCard{\V}$; consequently, we have a contradiction that $\SetCard{\V} < k$.
So the equilibrium cannot be that all nodes choose $1$.
Suppose in the equilibrium the nodes are divided into two sets $S$ and $T$ where nodes in $S$ choose $1$ and nodes in $T$ choose $0$. 
Further note that $t$ cannot be in $S$. Because another node in $S$ would prefer to choose $0$ as it connects to $t$. 
Thus, the nodes in $S$ form an independent set. 
For a node $u$ in $T$, it must be true that $u$ connects to at least one node in $S$ since $u$ chooses $0$. 
Furthermore, $t$ choosing $0$ means that the size of $S$ is at least $k$. Thus, the set $S$ is an independent
set with size at least $k$.
\end{proof}

Due to the hardness results presented in Theorem~\ref{thm-hard}, we next restrict the underlying graph to be either a complete graph or a tree, which leads to a series of positive results.

\subsection{Games on Complete Graphs}
We assume that the underlying graph \G is complete.
For ease of exposition, we omit the parameterization of $\mathcal{U}$ and an instance of BNPG is denoted by $\BNPG{\G}$.
We show that it takes polynomial time to find a PSNE or conclude no one exists.

First, we characterize the non-trivial PSNEs of a heterogeneous BNPG.

        \begin{proposition}\label{P:psne}
            Consider a heterogeneous \BNPG{\G} where \G is a complete graph.
            The game has a non-trivial PSNE with $k \in (0, n)$ investing players if and only if 
            1) there is no \Cost{i} such that $\Dg{i}{k-1} < \Cost{i} < \Dg{i}{k}$ and 
            2) $\SetCard{\I{+}(k)} \le k$, $\SetCard{\I{-}(k)} \le n-k$, and $\SetCard{\Set{i \in \V}{i \in \V \setminus (\I{+}(k) \cup \I{-}(k))}} \ge k - \SetCard{\I{+}(k)}$.
        \end{proposition}
        \begin{proof}
            $(\Rightarrow)$
            Suppose \ActionVec is a PSNE with $k \in (0, n)$ investing players.
            From Proposition~\ref{P:always}, for player $i$ if $\Cost{i} < \Dg{i}{k}$ the player must invest.
            Thus, the number of players in $\I{+}(k)$ is at most $k$; otherwise, there would be more than $k$ investing players.
            Similarly,  if $\Cost{i} > \Dg{i}{k-1}$  the player never invests; thus, the number of players in the set $\I{-}(k)$ is at most $n-k$; otherwise, the number of players that can invest would be less than $k$.
            From Proposition~\ref{P:always}, we know that there is no \Cost{i} such that $\Dg{i}{k-1} < \Cost{i} < \Dg{i}{k}$.
            Let $\T= \V \setminus (\I{+}(k) \cup \I{-}(k))$, i.e., the set of players who are indifferent between investing and not investing.
            To have exactly $k$ investing players, we need $k - \SetCard{\I{+}(k)}$ players from \T, which shows the lower bound on its cardinality.
            
            $(\Leftarrow)$
            We construct a PSNE with exactly $k$ investing players.
            First, let the players in $\I{+}(k)$ invest; they do not deviate as they always invest. 
            Then, we arbitrarily pick $k - \SetCard{\I{+}(k)}$ players from \T and make them invest; they do not deviate as they are indifferent between investing and not investing.
            Finally, let the remaining players not invest; they do not deviate as they are either indifferent between investing/not investing or never invest. 
        \end{proof}

Proposition~\ref{P:psne} suggests the following algorithm to compute a PSNE or conclude that on one exists.
Given \BNPG{\G}, we first make sure there is not any $\Cost{i}$ such that $\Dg{i}{k-1} < \Cost{i} < \Dg{i}{k}$; otherwise, we conclude that the game has no PSNE.
We then check if the game has any trivial PSNE.
Finally, by changing $k$ from $1$ to $n-1$, we check for each $k$  if the game has a non-trivial PSNE with $k$ investing players.
The algorithm runs in $O(n^2)$ time.

\subsection{Games on Trees}
Another important special case is when \G is a tree.
We now present a polynomial-time algorithm to compute a PSNE or
conclude no one exists. 
Our algorithm, \textsc{TreePSNE}, is inspired by
\textsc{TreeNash} proposed by \citet{kearns2013graphical}, but unlike the
latter, it computes an exact PSNE (rather than an approximate mixed
Nash equilibrium) in polynomial time.

\begin{theorem}\label{th:TreePSNE}
Given a heterogeneous \BNPG{\G} where \G is a tree.
It takes polynomial time to compute a PSNE or conclude that no one exists.
\end{theorem}

\begin{proof}
The proof is constructive.
Suppose \G is an inverted tree, where the root is at the bottom and the leaves at the top. 
The nodes in \G are categorized into three classes: leaf nodes, internal nodes, and the root. 
Our algorithm consists of two passes: a \emph{downstream} pass and an \emph{upstream} pass.
In the downstream pass we traverse the nodes in a depth-first order (start with the leaves and end at the root).
Each leaf or internal node passes a table to its parent; the table contains the best responses of the node conditioned on the action of its parent; we call the table \emph{conditional best-response table}.

\providecommand{\SubG}[1]{\ensuremath{\G_{#1}}\xspace}
\providecommand{\Inv}[1]{\ensuremath{ n'_{#1}  }\xspace}
\providecommand{\Y}{\ensuremath{k}\xspace}
\providecommand{\PY}{\ensuremath{p}\xspace}

Consider a node \Y.
Let \SubG{\Y} be the subtree consisting of \Y and its children; if \Y is a leaf node then the tree is a single node.
Define \Inv{\Y=1} (resp. \Inv{\Y=0})  as the number of \Y's  children that invest given that \Y invests (resp. does not invest); $\Inv{\Y=1} = \Inv{\Y=0} = 0$ if \Y is a leaf node.
For now suppose \Y is not the root; we will discuss the root case separately.
Let \PY be \Y's parent.
Denote their action profile by $(\Action{\PY}, \Action{\Y})$.
The necessary condition that $(\Action{\PY}, \Action{\Y})$ is part of a PSNE is captured in \eqref{eq:downstream}, where the ``$\diamond$'' is ``$\ge$'' if $\Action{\Y}=1$ and ``$\le$'' otherwise.

\begin{equation}\label{eq:downstream}
    \Dg{\Y}{\Inv{\Y=\Action{\Y}} + \Action{\PY}} \diamond \Cost{\Y}.
\end{equation}

Suppose \Y is a leaf node and \PY is its parent. 
To compute the best-response table of \Y, we exhaustively check the four possible action profiles of \Y and \PY, which takes $O(1)$ time. 
If no action profile satisfies \eqref{eq:downstream}, we conclude that no PSNE exists.

\begin{figure}[H]
\centering
\begin{tikzpicture}[thick,
  every node/.style={circle, scale=0.6},
  LeafNode/.style={draw,fill=myblue, minimum size=1pt},
  InternalNode/.style={draw,fill=mygreen, minimum size=1pt},
  every fit/.style={rectangle,draw,inner sep=-2pt,text width=2cm},
  ->,shorten >= 2pt,shorten <= 2pt
]

\node[InternalNode] (j1) {$j_1$};
\node[draw,rectangle,scale=1.2] at (-1.5, -1) {
	$\begin{aligned}
		& x_{\Y} = 1 \rightarrow x_{j_1} = 0 \\
		& x_{\Y} = 0 \rightarrow x_{j_1} = 1
	\end{aligned}$
};
\node[scale=1.5] at (1, 0) (tmp) {$\ldots$};
\node[InternalNode] at (2,0) (jm) {$j_m$};
\node[draw,rectangle,scale=1.2] at (3.8, -1) {
	$\begin{aligned}
		& x_{\Y} = 1 \rightarrow x_{j_m} = 1  \\
		& x_{\Y} = 1 \rightarrow x_{j_m} = 0  \\
		& x_{\Y} = 0 \rightarrow x_{j_m} = 1  \\
		& x_{\Y} = 0 \rightarrow x_{j_m} = 0  \\
	\end{aligned}$
};
\node[InternalNode] (Y) at (1, -1) {$\Y$};
\node[InternalNode] (W) at (1, -2) {$\PY$};

\draw[-](j1) -- (Y);
\draw[-](jm) -- (Y);
\draw[-](Y) -- (W);

\end{tikzpicture}

\small \caption{The conditional best-response table for internal nodes $j_1, \dots, j_m$.}
\label{fig:internalNode}
\end{figure}
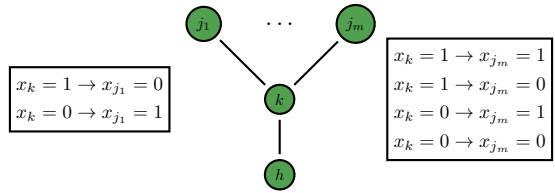

\providecommand{\AlwayInv}[1]{\ensuremath{ n^1_{#1} }\xspace}
\providecommand{\NeverInv}[1]{\ensuremath{ n^0_{#1} }\xspace}

Now, suppose \Y is an internal node, with $m$ children $j_1, \ldots, j_m$; the parent of \Y is \PY.
For the purpose of induction, suppose the conditional best-response tables of $j_1, \ldots, j_m$ have been passed to \Y.
An example is showed in Figure~\ref{fig:internalNode}, where the two tables beside $j_1$ and $j_m$ consist of the best responses of the two nodes conditioned on \Y's action.
We need to check if there is a an action profile $(\Action{\PY}, \Action{\Y})$ satisfying \eqref{eq:downstream}, which requires us to compute \Inv{\Y=0} and \Inv{\Y=1}.
Recall that \dMax is the maximum degree of \G.
When $\Action{\Y} = 0$, we denote the number of \Y's children that always invest (resp. never invest) as \AlwayInv{\Y=0} (resp. \NeverInv{\Y=0}).
To compute \AlwayInv{\Y=0}, we iterate through the conditional best-response tables of $j_1, \ldots, j_m$ and  count the number of children who always invest conditioned on $\Action{\Y}=0$.
For example, node $j_1$ in Figure~\ref{fig:internalNode} always invests when $\Action{\Y}=0$, since that is the only option; however, node $j_m$ does not necessarily invest since both $\Action{j_m}=0$ and $\Action{j_m}=1$ are best responses.
We similarly compute \NeverInv{\Y=0}.
Note that $\AlwayInv{\Y=0} + \NeverInv{\Y=0} \le m$.
Thus, the range of $\Inv{\Y=0}$ is $\SET{ \AlwayInv{\Y=0}, \AlwayInv{\Y=0} + 1, \ldots,  m - \NeverInv{\Y=0} }$.
We obtain the range of $\Inv{\Y=1}$ similarly.
Finally, we identify action profiles of \PY and \Y that are part of a PSNE,  by substituting into \eqref{eq:downstream} the four possibilities of $(\Action{\PY}, \Action{\Y})$ and the corresponding values of $\Inv{\Y}$ conditioned on \Action{\Y}.
If no $(\Action{\PY}, \Action{\Y})$ satisfies \eqref{eq:downstream}, we conclude that  no PSNE exists.

Suppose node \Y is the root.
We compute the \emph{best-response table} of \Y as follows.
Following the above procedure for an internal node, we compute \Inv{\Y=0} and \Inv{\Y=1}.
Then, we check if $\Action{\Y}=0$ (resp. $\Action{\Y}=1$) is a best response by checking if $\Dg{\Y}{\Inv{\Y=0}} \le \Cost{\Y}$ (resp. $\Dg{\Y}{\Inv{\Y=0}} \ge \Cost{\Y}$) is satisfied.
Finally, we put \Y's best responses into the table. 
In summary, the downstream pass takes $O(\dMax \SetCard{\V} + \SetCard{\E})$ time.

We show that the downstream pass does not miss any PSNE if there is one, by an inductive argument. 
Suppose there is a PSNE that is missed by the downstream pass.
Then we must wrongly conclude that no PSNE exists at some node \Y as we move downward. 
First, the node \Y cannot be a leaf node, since we exhaustively check the four possible action profiles of \Y and its parent. 
Next, suppose the PSNE is not missed until an internal node \Y, in other words, we must wrongly exclude an action profile of \Y and its parent \PY.
However, we exhaustively check all combinations of $\Inv{\Y}$, $\Action{\Y}$, and $\Action{\PY}$, so it is not possible to miss such an action profile.
The same argument applies for the case where node $k$ is the root.

In the upstream pass we traverse \G in a reversed depth-first order (start at the root and end the leaves). 
We first select the action of the root.
If the root is indifferent between investing and not investing, we arbitrarily pick one. 
Next, we sequentially determine each node's action based on its parent's action and the conditional best-response table passed to the parent (ties are broken arbitrarily).
This pass takes $O(\SetCard{\V})$ time.
In conclusion, the total running time of \textsc{TreePSNE}
is $O(\dMax \SetCard{\V} + \SetCard{\E})$.
\end{proof}

\subsection{Games on Arbitrary Graphs}

We now describe a heuristic algorithm to find an approximate PSNE for $\BNPG{\G}$ defined on a general graph.
An approximate PSNE is defined as follows:

\begin{defn}
    Given a BNPG, an approximate PSNE (i.e., \aPSNE) is an action profile $\ActionVec \in \SET{0, 1}^n$ such that for all $i \in \V$:
	    \begin{equation*}
	           \Util{i}(\Action{i}, \InvNum{i}{\ActionVec}) + \epsilon \ge  \Util{i}(1-\Action{i}, \InvNum{i}{\ActionVec}).
	    \end{equation*}
\end{defn}

The heuristic algorithm consists of two subroutines.
The first one is termed as \textsc{Asynchronous-BR},
which is based on the idea of best-response dynamics.
In each iteration of \textsc{Asynchronous-BR}, every player sequentially updates her action by best-responding to the current action profile in an asynchronous manner (i.e., later players can observe the actions of former players). 
The second subroutine is termed \textsc{Evolve}; it executes \textsc{Asynchronous-BR} \numEvo times and each execution generates an \aPSNE; finally \textsc{Evolve} picks the \aPSNE with the minimum $\epsilon$.
The heuristic algorithm is summarized in Algorithm~\ref{algo:heuristic}, where the two sub-routines are described in Algorithms~\ref{algo:asyn-br} and \ref{algo:evolve}.
The input \StopCond is the stopping criterion; the algorithm stops when the distance of two consecutive profiles (e.g., measured by some $\ell_p$ norm) is less than \StopCond or when the maximum number of iterations is reached.

\begin{algorithm}[ht]
	\caption{Heuristic to find \aPSNE}\label{algo:heuristic}
		\begin{algorithmic}[1]
			\State \textbf{Input}: $\numEvo$, $\StopCond$, $K$ \Comment{$B$: maximum number of iterations}
			\State Initialize: $d = M, i = 0$ \Comment{$M$: a large positive number}
			\State $\ActionVec \leftarrow $ random initializatioin
			\If{\ActionVec is a PSNE}
			    \State return \ActionVec
			\EndIf
			\While{$d \ge \StopCond$ and $i < B$}
			\State $i \leftarrow i + 1$
			\State $\ActionVec' \leftarrow$ \textsc{Evolve}($\ActionVec, \numEvo$) \Comment{run $\numEvo$ trials of the \textsc{Evolve} subroutine}
			\If{$\ActionVec'$ is a PSNE}
				\State return $\ActionVec'$
			\Else
				\State $d \leftarrow \Norm{\ActionVec' - \ActionVec}_p$ 
			\EndIf
			\EndWhile
			\State return $\ActionVec'$
		\end{algorithmic}
\end{algorithm}

\begin{algorithm}[ht]
	\caption{Asynchronous-BR}\label{algo:asyn-br}
	\begin{algorithmic}[1]
		\State \textbf{Input}: $\ActionVec$
		\For{$i=1, \ldots, n$}
		\If{$\Delta g_i(n_i) > c_i$}
		    \State $\Action{i} = 1$
		\ElsIf{$\Dg{i}{n_i} < c_i$}
		    \State $\Action{i} = 0$
		\Else
		    \State randomly set $\Action{i}=1$ or $0$ with probability $0.5$
		\EndIf
		\EndFor
	\end{algorithmic}
\end{algorithm}

\begin{algorithm}[ht]
	\caption{Evolve}\label{algo:evolve}
	\begin{algorithmic}[1]
		\State \textbf{Input}: $\ActionVec, \numEvo$
		\State Initialize: $\epsilon^\ast=M, \ActionVec^\ast=\bm{0}$ \Comment{$M$: a large positive number}
		\For{$i=1, \ldots, \numEvo$}
		\State $\epsilon \leftarrow$ \Call{maxEpsilon}{$\bm{x}$} \Comment{\Call{maxEpsilon}{\ActionVec}: $\max\Set{\Util{i}(1-\Action{i}, \InvNum{i}{\ActionVec}) - \Util{i}( \Action{i}, \InvNum{i}{\ActionVec})}{i \in \V}$}
		\If{$\epsilon = 0$}
		    \State return \ActionVec \Comment{Find a PSNE}
		\EndIf
		\If{$\epsilon < \epsilon^\ast$}
		    \State $\epsilon^\ast = \epsilon, \ActionVec^\ast=\ActionVec$
		\EndIf
		\State $\bm{x} \leftarrow$ \texttt{Asynchronous-BR}($\bm{x}$)
		\EndFor
		\State return $\ActionVec^\ast$
	\end{algorithmic}
\end{algorithm}

\section{Homogeneous BNPG}

In this section we study algorithmic issues of a \emph{homogeneous} BNPG, i.e., all players share the same externality function $g$.
We construct an example to show that a PSNE may not exist.
In addition, we show that checking the existence of a non-trivial PSNE is NP-complete.
We then make further restrictions on either the underlying graph \G or the utility functions.
Specifically, when \G is a complete graph we show that it takes polynomial time to compute a socially optimal PSNE.
In addition, when the players share the same externality function $g$ and the same cost $c$, the non-trivial PSNEs are characterized by $k$-cores of \G.
In the following discussion, we simplify \InvNum{i}{\ActionVec} to $n_i$ when the context is clear.

\subsection{Existence of PSNE}\label{sec:homo-BNPG-existence}

We construct a game where a PSNE does not exist. 
This game has three players $\SET{1,2,3}$. 
The underlying graph is a simple path with player $2$ in the middle. 
We set the costs as follows: $\Cost{1} = 1$, $\Cost{2} = 2$, and $\Cost{3} = 3$. 
The externality function satisfies that $\Dgh{0} = 1.5$, $\Dgh{1} = 3.5$,  and $\Dgh{2} = 0.5$; one way to realize such a function is $g(0)=4.5$, $g(1)=6$, $g(2)=9.5$, and $g(x)=10$ for any $x \ge 3$, ensuring that $g$ is non-decreasing.
By exhaustively checking the eight possible profiles, it is direct to verify that no PSNE exists.

Recall that a trivial PSNE is an action profile where either no player invests or all players invest.
Given a homogeneous BNPG, it is direct to check if the game has a trivial PSNE.
A non-trivial PSNE is an action profile where exactly $k$ players invest with $0 < k < n$.
The following theorem shows that checking the existence of a non-trivial PSNE is NP-complete, which is implied from Theorem~\ref{th:fully-homo-hard}.

\begin{theorem}
\label{th:homo-hard}
Checking if a homogeneous BNPG has a non-trivial PSNE is NP-complete.
\end{theorem}

Next, we consider further restrictions that enable positive algorithmic results.
In particular, we consider two cases:  1) the underlying graph \G is a complete graph and  1) in addition to sharing the same externality function, the players share the same cost $c$.

\subsection{Games on Complete Graphs}

\subsubsection{Computing a PSNE} 

Given an instance \BNPG{\G} where \G is a complete graph, we show that there always exists a PSNE and it can be computed in $O(n \log n)$ time.
Recall from Proposition~\ref{P:psne}, it takes $O(n^2)$ time to compute (or conclude no one exists) a PSNE of a heterogeneous $\BNPG{\G}$ with a complete \G; the homogeneity improves the algorithm to $O(n \log n)$.

\begin{theorem}\label{th:existence}
	Every homogeneous BNPG on a complete graph has a PSNE, which can be computed in  $O(n \log n)$ time.
\end{theorem}
\begin{proof}
The proof is constructive.
Without loss of generality, suppose $c_1 \leq c_2 \leq \cdots \leq c_n$ (we can start by sorting players in this order).
If $\Delta g(0) < c_1$, we immediately obtain a PSNE, i.e., no player invests; this is because for $i=1,\ldots, \NumAgent$ we have $\Util{i}(0, n_i) > \Util{i}(1, n_i)$ where $n_i=0$.
If $\Delta g(n-1) > c_\NumAgent$, the action profile that all players invest is a PSNE, since for $i=1, \ldots, \NumAgent$ we have $\Util{i}(1, n_i) > \Util{i}(0, n_i)$ where $n_i = \NumAgent-1$.
If neither of the  two cases holds, we initialize $\bm{x}$ to all-zeros. Since $c_1\leq \Delta g(0)$, we set $x_1 = 1$.  For a subsequent player $i$, we set $x_i = 1$ if $c_i \leq \Delta g(n_i)$. Note that at this point $n_i$ is equal to the number of players with index smaller than $i$. We repeat this process for each player in ascending order of $c_i$ until we found a certain player $m$ such that $c_m> \Delta g(n_m)$. We then set the decisions of player $m$ and the subsequent players to $0$. 
We show that the resulting action profile is a PSNE.

First, there are $m-1$ investing players in the PSNE.
For player $m-1$, we have $c_{m-1} \leq  \Delta g(n_{m-1})$ by construction, so player $m-1$ does not deviate from investing (also note that $n_{m-1} = m-2$).
For player $m$, we have $c_m> \Delta g(n_{m})$ by construction, so player $m$ keeps not investing (also note that $n_m = m-1$).
For any player $i$ with $i < m-1$,  $n_i$ is equal to $(i-1) + (m-1-i) = m-2$.
The player does not deviate from investing since $c_i \le c_{m-1} \leq \Delta g(n_{m-1}) = \Delta g(m-2) = \Delta g(n_i)$.
For any player $j$ with $j > m$, $n_j$ is equal to $m-1$.
The player keeps not investing since $c_j \ge c_m > \Delta g(n_m) = \Delta g(m-1) = \Delta g(n_j)$.
The time complexity of finding a PSNE is dominated by the $O(n \log n)$ sorting of $\Cost{i}$.
\end{proof}

\begin{algorithm}[ht]
	\caption{\textsc{SimpleSort}}\label{algo:simplesort}
	\begin{algorithmic}[1]
		\State \textbf{Input}: a homogeneous \BNPG{\G} with externality function $g$ and cost parameters $\Cost{1}, \ldots, \Cost{\NumAgent}$.
		\State Sort the costs such that $\Cost{1} \le \cdots \le \Cost{\NumAgent}$
		\If{$\Delta g(0) < \Cost{1}$}
		    \State return $\bm{x}=\bm{0}$ as a PSNE
		\EndIf
		\If{$\Delta g(\NumAgent-1) > \Cost{\NumAgent}$}
		    \State return $\bm{x}=\bm{1}$ as a PSNE
		\EndIf
		\State Initialize $\bm{x}=\bm{0}, n=0$
		\For{$i=1,\ldots, \NumAgent$}
		    \If{$\Cost{i} \le \Delta g(n)$}
		        \State $x_i = 1$
		        \State $n = n + 1$
		    \Else
		        \State $x_i = x_{i+1}, \ldots x_{\NumAgent} = 0$
		        \State break
		    \EndIf
		\EndFor
	    \State return $\bm{x}$
	\end{algorithmic}
\end{algorithm}

We summarize the proof of Theorem~\ref{th:existence} as in Algorithm~\ref{algo:simplesort}.
The algorithm outputs a PSNE with a special structure: the first $k$ players (in ascending order of $c_i$)  invest and the rest do not invest, for some $k$ determined by the algorithm. 
Notice that the output PSNE may not be unique, as some of the first $k$ players may be indifferent between investing and not investing.

\subsubsection{Computing Socially Optimal PSNE}
Given an action profile $\bm{x}$, the social welfare (SW) is defined as the sum of all players' utilities: 
\begin{equation}\label{eq:sw}
\text{SW}(\bm{x}) := \sum_{i=1}^{\NumAgent}{\Util{i}(x_i, n_i)}.
\end{equation}
A common goal is to seek the PSNE that maximizes social welfare---a \emph{socially optimal} PSNE. 
Although the \textsc{SimpleSort} algorithm (i.e., Algorithm~\ref{algo:simplesort}) always finds a PSNE, it is not necessarily a socially optimal one.
We show in Theorem~\ref{th:social-opt} that the socially optimal PSNE can be efficiently computed.

\begin{theorem}\label{th:social-opt}
    Consider a homogeneous \BNPG{\G} where \G is a complete graph, it takes polynomial time to compute the socially optimal PSNE.
\end{theorem}
\begin{proof}
It is direct to check if the game has any trivial PSNE and if so to compute the corresponding SW.
For a particular $0 < k < \NumAgent$, we call a PSNE with $k$ investing players as a $k$-PSNE.
Given a $k$-PSNE, we know that $\Action{i} + n_i = k$ for all $i$ since the underlying graph is complete.
The social welfare of a $k$-PSNE is re-written as follows
    \begin{equation}\label{eq:sw-k}
        \text{SW}(\ActionVec; k) = n \cdot g(k) - \sum_{i=1}^n  c_i \Action{i}.
    \end{equation} 
Next, we describe how to compute the socially optimal $k$-PSNE (if it exists): 
\begin{itemize}
    \item Case 1: When $\Delta g(k) > \Delta g(k-1)$,  if there is a $\Cost{i}$ such that $\Delta g(k-1) < \Cost{i} < \Delta g(k)$ then the $k$-PSNE does not exist.
    Otherwise, let $\I{+}(k) = \Set{i}{\Cost{i} \le \Delta g(k-1)}$.
    By Proposition~\ref{P:always}, the $k$-PSNE exists and is unique if and only if $\SetCard{\I{+}(k)}=k$.
    If the $k$-PSNE exists, it is the socially optimal $k$-PSNE.
    
    \item Case 2: When $\Delta g(k) \leq \Delta g(k-1)$, let $\I{\times}(k):=\V \setminus \I{-}(k)$, i.e., the players that either always invest or are indifferent between investing and not investing.
    The $k$-PSNE exists if and only if $\SetCard{\I{\times}(k)} \ge k$.
    When the $k$-PSNE exists, the social welfare is maximized if we choose the $k$ investing players as those with the smallest \Cost{i}, which can be identified by a simple sorting of $\I{\times}(k)$ with respect to \Cost{i}.
\end{itemize}
In summary, an algorithm to compute a socially optimal PSNE is as follows:
we compute the social welfare of each $k$-PSNE (if it exists) from $k=1$ to $k=n-1$.
Then,if the game has any trivial PSNE, we compute the corresponding social welfare.
The socially optimal PSNE is obtained by comparing the above PSNEs with respect to their SW.
\end{proof}

\subsection{Fully Homogeneous BNPG}

A homogeneous BNPG is called fully homogeneous if the players share the same cost $c$.
Further, we restrict the function $g$ to be strictly convex; as a result, a PSNE always exists and can be efficiently computed.
The convexity of $g$ is essential for the positive results; without the convexity checking if a fully homogeneous BNPG has a non-trivial PSNE is NP-complete.

\subsubsection{Strictly Convex $g$}
Given a fully homogeneous $\BNPG{\G}$,  let \dMax be the maximum degree of the underlying graph \G.
Since $g$ is strictly convex, the first-order derivative is strictly increasing, i.e., $\Delta g(0) < \Delta g(1) < \cdots < \Delta g(\dMax)$. 
The following proposition characterizes the extreme case where a trivial PSNE is the only Nash equilibrium:
    \begin{proposition}\label{prop:fully-and-convex}
        If $c < \Delta g(0)$ (resp. $c > \Delta g(\dMax)$), the action profile that all players invest (resp. no one invests) is the unique PSNE.
    \end{proposition}
    \begin{proof}
        Consider $c < \Delta g(0)$. 
        Suppose there is an action profile that some player $i$ does not invest; this implies that $c \ge \Delta g(n_i)$ with $n_i \ge 0$; it follows that $c \ge \Delta g(n_i) \ge \Delta g(0)$, which leads to a contradiction.
        Next, consider $c > \Delta g(\dMax)$.
        Suppose there is an action profile that some player $i$ invests, which implies that $c \le \Delta g(n_i)$.
        As $n_i$ is at most \dMax, it follows that $c \le \Delta g(\dMax)$, which leads to a contradiction.
    \end{proof}

Next, we connect non-trivial PSNEs of the game with $k$-cores of the underlying graph \G.
Specifically, a $k$-core of a graph is a \emph{maximal} induced subgraph such that each node in the subgraph has degree at least $k$; we follow the definition as in~\cite{manshadi2009supermodular} (Definition 4.1), which does not require a $k$-core to be connected.

\begin{theorem}
	\label{thm-fully-homo-PSNE}
	Given a fully homogeneous BNPG  with strictly convex $g$. 
	Suppose  $\ActionVec \notin \SET{\bm{0}, \bm{1}}$.
	Let $\mathcal{H}$ be the induced subgraph over the investing players.
	\ActionVec is a non-trivial PSNE if and only if $\mathcal{H}$ is a $k$-core of the underlying graph with $k = \Delta g^{-1}(c) \ge 0$. 
\end{theorem}

\begin{proof}
    ($\Rightarrow$)
	Let $\mathcal{I}_1$ be the (non-empty) set of investing players, i.e., $\mathcal{I}_1=\Set{i}{ \Action{i} = 1}$.
	For any player $i \in \mathcal{I}_1$, it follows that $c \le \Delta g(n_i)$, or equivalently, $\Delta g^{-1}(c) \le n_i$, where $\Delta g^{-1}(\cdot)$ is the inverse function of $\Delta g$; it is well-defined due to the monotonicity of $\Delta g$. 
	As $\mathcal{H}$ is the induced subgraph on the investing players, the degree  of node $i \in \mathcal{H}$ satisfies that $d_{\mathcal{H}}(i)=n_i \ge \Delta g^{-1}(c)$.
	Let $k = \Delta g^{-1}(c)$.
	We show that $k \ge 0$: suppose $k < 0$, it follows that $\Delta g^{-1}(c) < 0$ and hence $c < \Delta g(0)$.
	By Proposition~\ref{prop:fully-and-convex}, $c < \Delta g(0)$ implies that the only PSNE is a trivial one, which contradicts that \ActionVec is non-trivial.
	
	($\Leftarrow$)
	Consider the action profile $\bm{x}$ where the players in $\mathcal{H}$ invest while the others do not invest. 
	We show that $\bm{x}$ is a PSNE.
	For any player $i \in \mathcal{H}$, its degree $d_\mathcal{H}(i) = n_i \ge \Delta g^{-1}(c)$; thus, we have $\Delta g(n_i) \ge c$ and the player does not deviate from investing.
	For any player $j \notin \mathcal{H}$, its degree is less than $k$; otherwise it would be in $\mathcal{H}$; this implies that $n_j \le d_\G(j) < k = \Delta g^{-1}(c)$; it follows that $c > \Delta g(n_j)$, i.e., the player keeps not investing.
\end{proof}

As $\Delta g^{-1}(c)$ is readily computable, we apply a simple pruning algorithm to find the $k$-core with $k=\Delta g^{-1}(c)$ in polynomial time (or conclude no one exists); this implies that we can find a non-trivial PSNE (or conclude no one exists) in polynomial time.
The pruning algorithm iteratively removes the nodes with degree less than $k$ until no more node can be removed. 
If the remaining graph is not empty, it is the $k$-core and the node set consists of the investing players; otherwise, the $k$-core does not exist and the game does not have any non-trivial PSNE.
It is clearly that the pruning algorithm runs in polynomial time.

\begin{corollary}
	Given a fully homogeneous BNPG with strictly convex $g$.
	When $\Delta g(0) \le c \le \Delta g(\dMax)$, a trivial PSNE (i.e., $\ActionVec = \bm{0}$) always exists. Besides, it takes polynomial time to find a non-trivial PSNE or conclude that no one exists. 
\end{corollary}
\begin{proof}
    $\ActionVec = \bm{0}$ implies that $n_i = 0$ for all $i$; the action profile is a PSNE since $\Util{i}(0, n_i) \ge \Util{i}(1, n_i)$ for all $i$.
    The remaining statement is direct to show from the discussion above.
\end{proof}

Finally, Theorem~\ref{th:fully-homo-hard} shows that the convexity of $g$ is necessary to
efficiently checking the existence of a non-trivial PSNE; intuitively, this is because the convexity leads to the monotonicity of $\Delta g$, which is essential to characterize the non-trivial PSNE by the $k$-core.

\begin{theorem}[\citet{yang2020refined}]\label{th:fully-homo-hard}
	\label{existence-non-trivial-PSNE}
	Consider a fully homogeneous BNPG with general $g$,  checking the existence of a non-trivial PSNE is NP-complete.
\end{theorem}
    

\section{Experiments}
In this section we conducted experiments to show the effectiveness of Algorithm~\ref{algo:heuristic} for finding approximate PSNE in general BNPGs.
The parameters $c_i$ are sampled from the uniform distribution on $[0, 1]$. 
In our experiments we assume $g_i = \lambda_i h_i$ where $\lambda_i$ is sampled uniformly at random on $[0,1]$, and $h_i$ are either convex or concave,
corresponding to strategic substitutes and complements, respectively~\cite{galeotti2010network}.
If $g_i$ is convex, it is sampled from a set of convex functions $ \{ -\alpha \log{(x+1)^\beta} \}$ by uniformly at random sampling $\alpha$ and $\beta$ from $\{0.1, 0.3, 0.5, 0.7, 0.9 \}$ and $\{1.2,  1.5, 2.0 \}$, respectively. Similarly, if $g_i(x)$ is concave it is sampled from a set of concave functions $\{-\alpha x^\beta \}$ by sampling $\alpha$ and $\beta$ from the same sets.  
We normalize the values of the utility functions $U_i$ to $[0,1]$ when evaluating the approximation quality of an $\epsilon$-PSNE. 
In each simulated game, we feature a mix of players with concave and convex $g_i$.
This mix is determined using a parameter $\gamma \in [0,1]$, which is the probability that a particular player $i$'s $g_i$ is concave; consequently, higher $\gamma$ implies a larger fraction of the population with convex utility functions.
For each value of $\gamma$ we simulated $1000$ BNPGs. 
We first conduct experiments on a Facebook network~\cite{leskovec2012learning}, which is a connected graph with $4093$ nodes and $88234$ edges.
Then, we experiment on synthetic networks to study the impact of network topologies on the fraction of players investing in PSNE.

\begin{figure}[h]
\centering
\setlength{\tabcolsep}{0.1pt}
\begin{tabular}{c}
\includegraphics[width=2.2in]{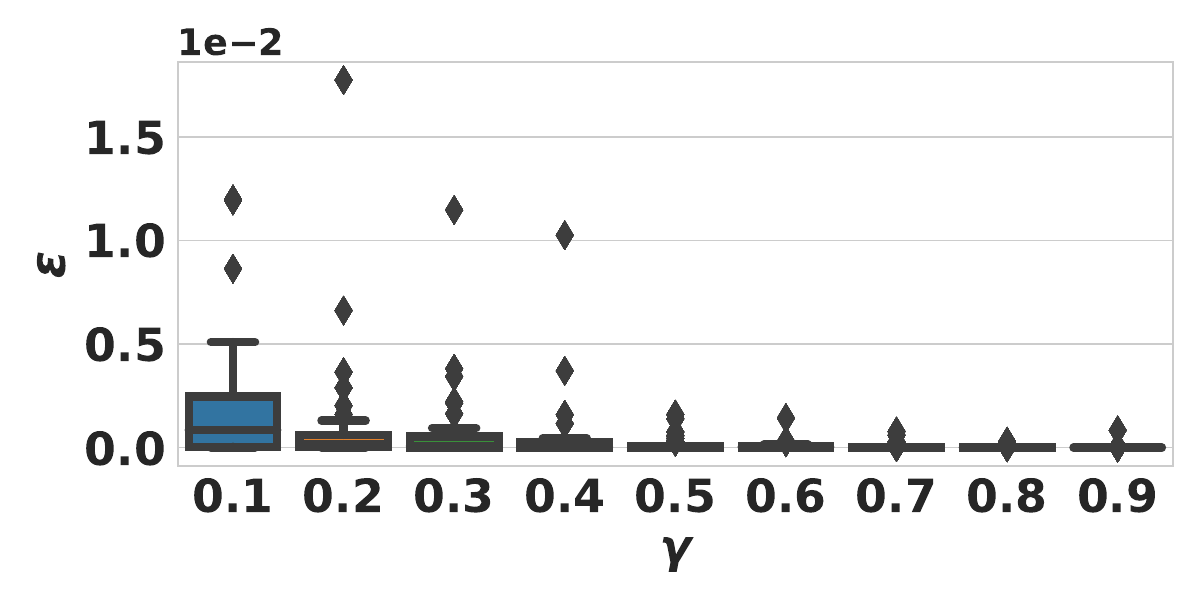}
\end{tabular}

\small \caption{Averaged $\epsilon$ of the approximate PSNE (in \% units).}
\label{fig:hete}
\end{figure}

\paragraph{Facebook Network.} The experimental results are shown in Figure~\ref{fig:hete}, where we vary $\gamma$.
Each bar is the averaged $\epsilon$ of the approximate PSNE computed using Algorithm~\ref{algo:heuristic}.
We omit the two corner cases of $\gamma=0$ and $\gamma=1$ because in all of these instances our algorithm returned an \emph{exact} PSNE.
The main takeaway from these results is two-fold: first, that even when populations are mixed so that a PSNE is not guaranteed to exist, there is usually an approximate PSNE with a small $\epsilon$ (maximum benefit from deviation), and second, that our heuristic algorithm finds good approximate PSNE (observe that even outliers have $\epsilon < 0.02$).

\begin{figure}[h]
\centering
\begin{tabular}{cc}
\includegraphics[width=0.45\columnwidth]{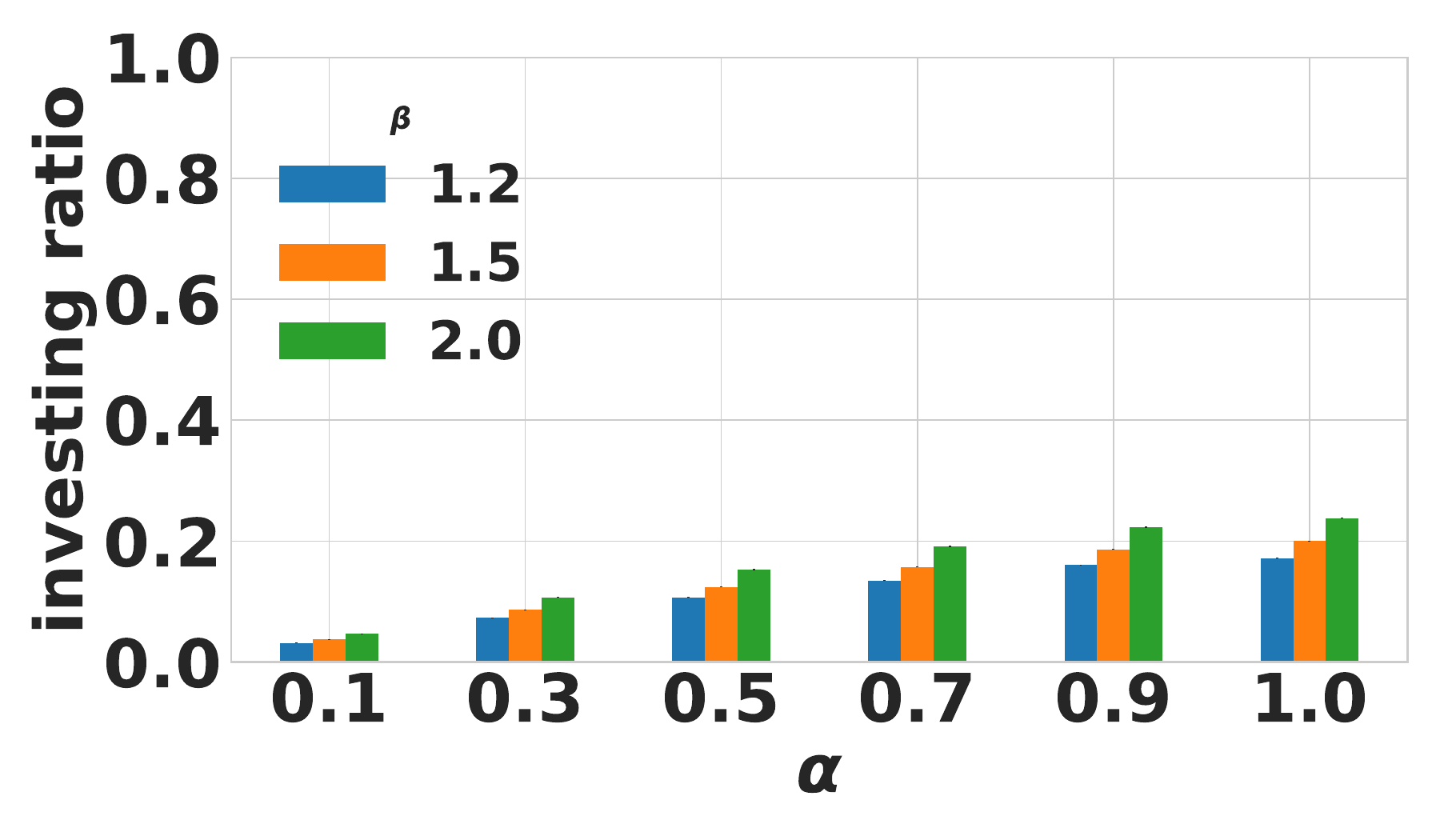} & \includegraphics[width=0.45\columnwidth]{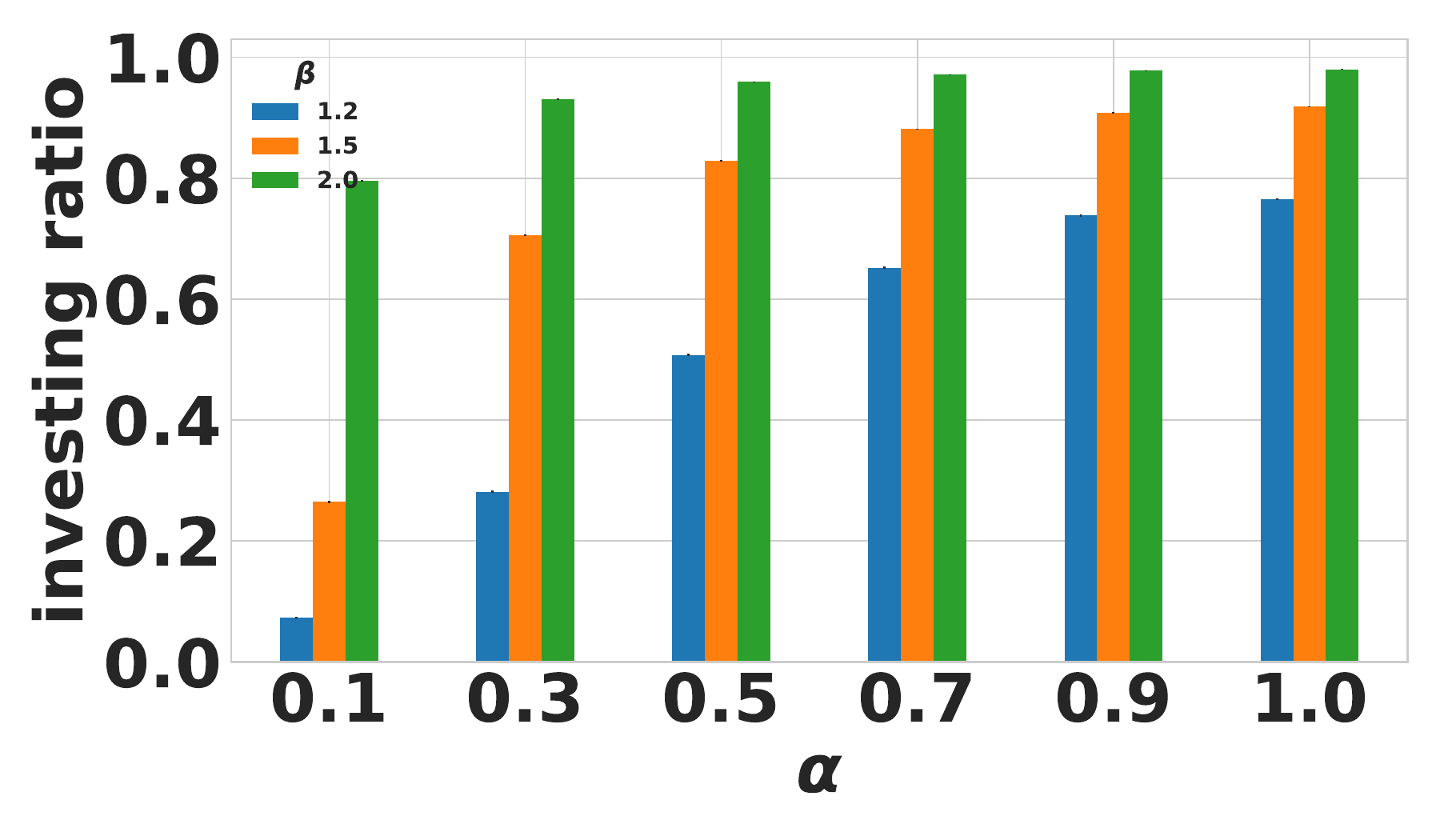} 
\end{tabular}

\small \caption{Ratio of investing players when $\gamma=0$ (left) and $\gamma = 1$ (right).}
\label{fig:investratio}
\end{figure}

Next, we study the impact of our three parameters, $\alpha$, $\beta$, and $\gamma$, on the fraction of players investing in PSNE.
Figure~\ref{fig:investratio} presents the results for $\gamma = 0$ (left) and $\gamma = 1$ (right) (these are the two cases where we can compute the exact PSNE, as mentioned above).
We observe that in both cases, increasing $\alpha$ or $\beta$ leads to increasing fraction of investing players.
This is because higher values of either increase the rate of change of the utility function as more players invest, and consequently network effects lead to higher overall investment.
The most dramatic difference, however, is between $\gamma = 0$ and $\gamma = 1$.
In the former case, a relatively small fraction ultimately invest, whereas $\gamma = 1$ leads to a great deal of equilibrium investment, particularly for sufficiently large $\alpha$ and $\beta$.
This is because convex $g_i$ imply that as more players invest, the marginal benefits to investment increase, and therefore equilibrium may be seen as a cascade of investment decisions that ultimately captures most of the community.

\paragraph{Synthetic Networks.} We conduct experiments on two types of synthetic networks. 
The first type is the Barabasi-Albert network (BA)~\cite{barabasi2009scale}. 
BA is characterized by its power-law degree distribution, where connectivity is heavily skewed towards high-degree nodes. 
The power-law degree distribution, $P(k)\sim k^{-r}$, gives the probability that a randomly selected node has $k$ neighbors. 
We consider three variants of the BA network with different exponents $r$.
We also experiment on Small-World (SW) network~\cite{watts1998collective}.
The Small-World network is well-known for balancing shortest path distance between pairs of nodes and local clustering in a way as to qualitatively resemble real networks. 
We consider three variants of the SW network, where they differ in the local clustering coefficients.
The details about BA and SW networks are listed in Table~\ref{tab:synthetic-network}.
We use the same experimental setup as the one for Facebook network. 
Note that we only consider $\gamma=0$ and $\gamma=1$, where a PSNE can always be found by Algorithm~\ref{algo:heuristic}.
When $\gamma=1$, a player has higher incentive to invest as more neighbors invest, while the incentive structure is the opposite when $\gamma=0$.


\begin{table}[h]
\centering
\small
\begin{tabular}{@{}cccc@{}}
\toprule
& BA-1     & BA-2     & BA-3     \\ \midrule
exponent $r$        & $2.7167$ & $2.2789$ & $2.0374$ \\
\cmidrule(r){1-4}
         			    &SW-1     & SW-2     & SW-3     \\ \midrule
local clustering coeff. & $0.3667$ & $0.3893$ & $0.4070$ \\ \midrule
\end{tabular}
\caption{Details about the BA and SW networks.}
\label{tab:synthetic-network}
\end{table}

The experimental results on BA networks are showed in Figure~\ref{fig:BA_investratio}. 
The columns from left to right correspond to results on BA-1, BA-2, and BA-3.
Note that as $r$ decreases it is more likely to have high-degree nodes. 
The top row (resp. bottom row) is the case where $\gamma=0$ (resp. $\gamma=1$).
At the bottom row, when $\alpha$ and $\beta$ are fixed, there are more high-degree nodes as $r$ decreases, thus the investing decisions of these high-degree nodes can encourage more nodes to invest, which results in an overall increasing trend to invest. 
The results at the top row have the opposite trend, where  more free riders occur due to the increase of high-degree nodes, which leads to an overall decreasing trend to invest. 
The experimental results for SW networks are showed in Figure~\ref{fig:Small-World_investratio}. 
The columns from left to right correspond to results on SW-1, SW-2, and SW-3.
Note that a larger local clustering coefficient leads to denser local structures.
Intuitively, when $\gamma=1$, an investing player's decision can encourage more neighbors to invest as the local structure becomes denser.
On the other hand, when $\gamma=0$, a denser local structure can also lead to more free riders.
This intuition is supported by the overall trend exhibited in Figure~\ref{fig:Small-World_investratio}.

\begin{figure}[h]
\centering
\begin{tabular}{lll}
\includegraphics[width=0.3\columnwidth]{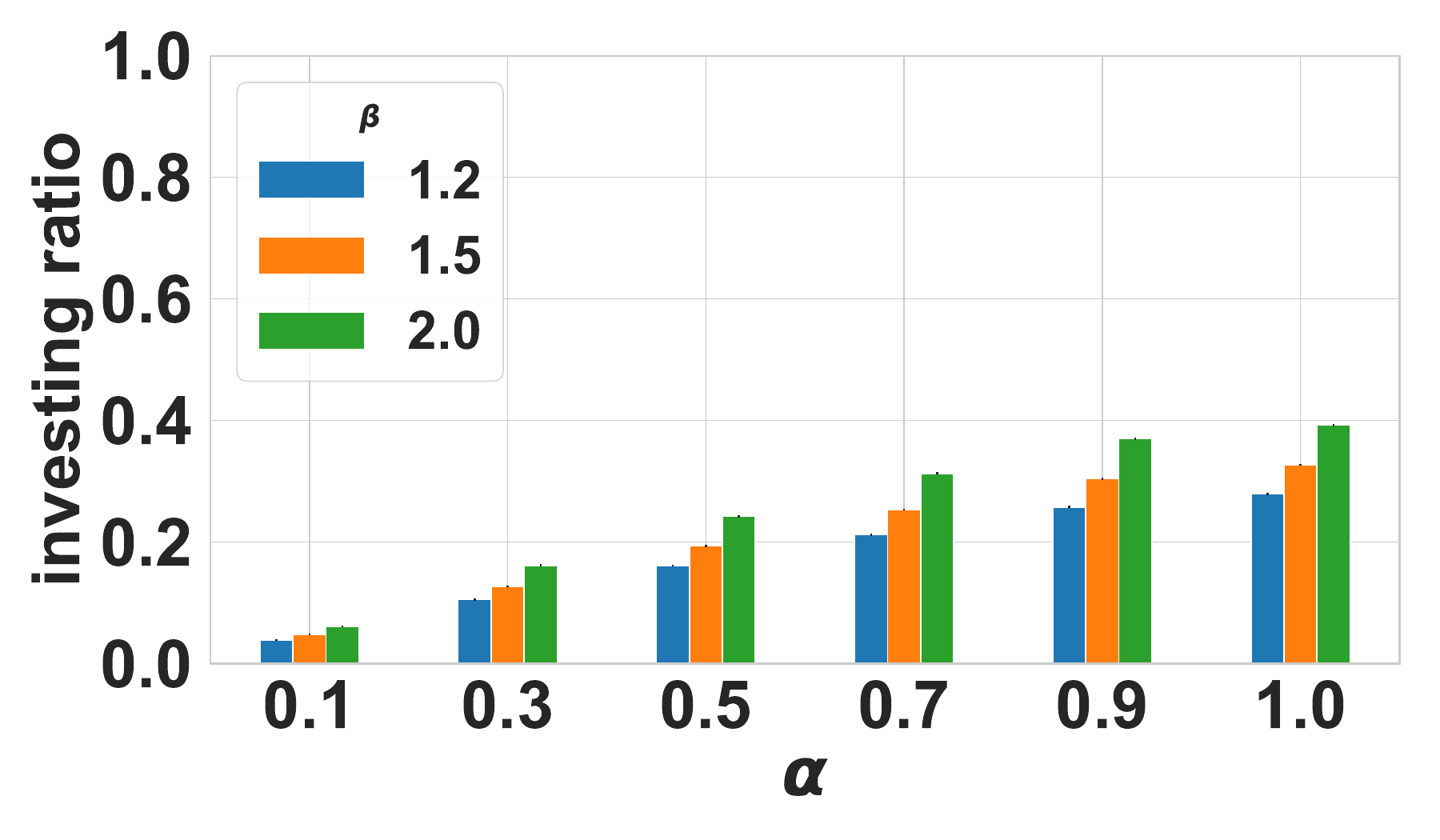} & \includegraphics[width=0.3\columnwidth]{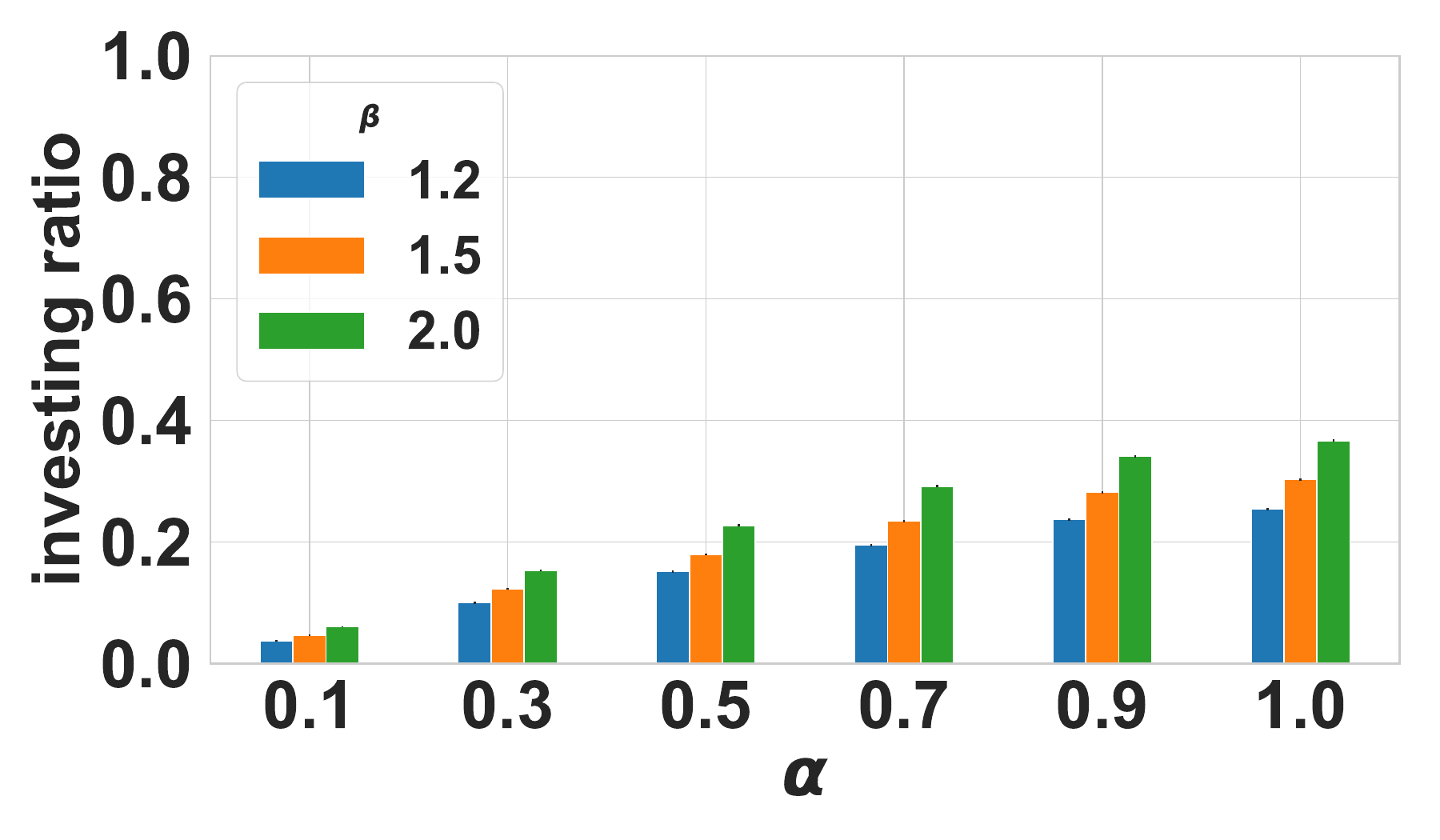} & \includegraphics[width=0.3\columnwidth]{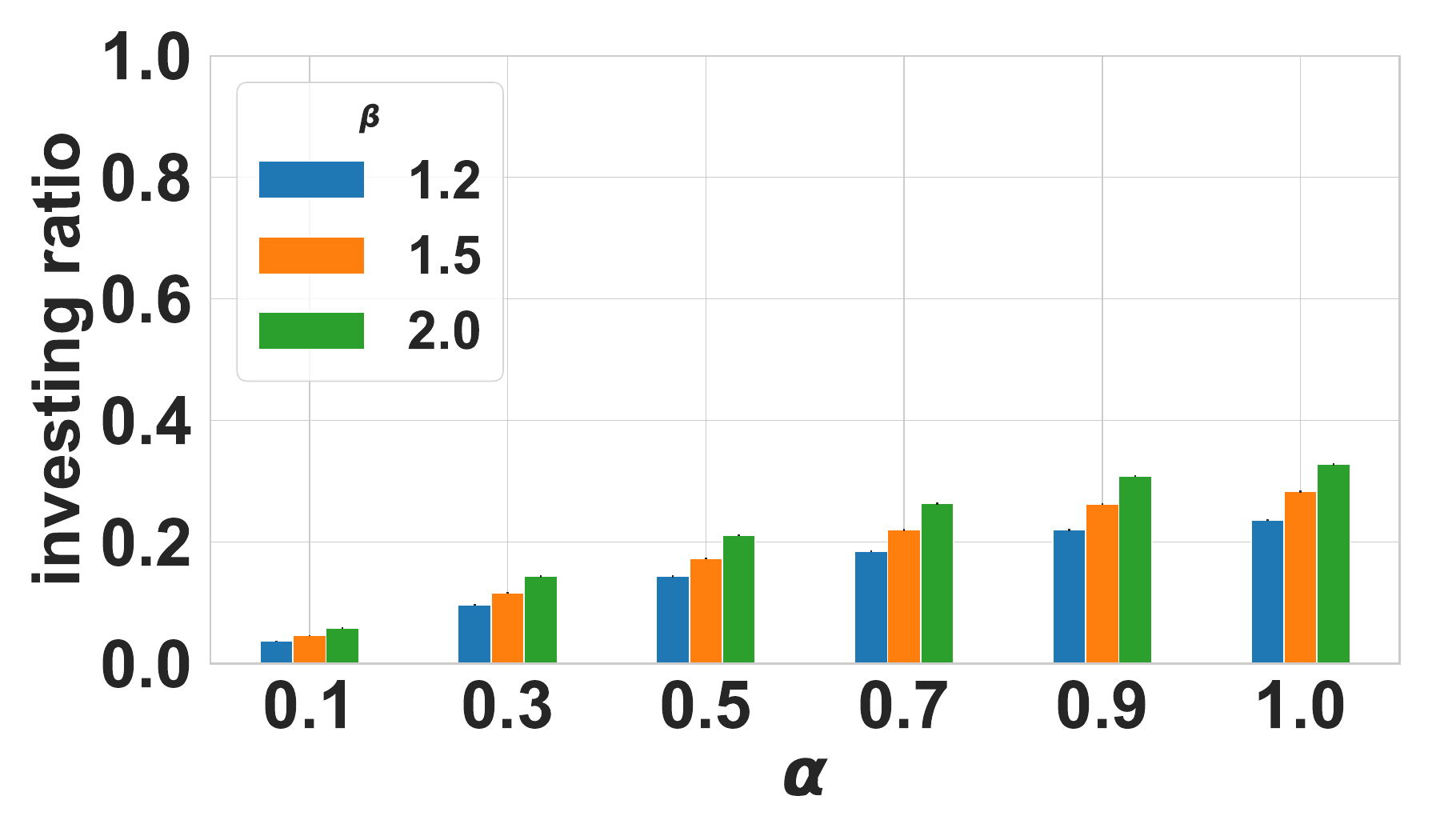} \\
\includegraphics[width=0.3\columnwidth]{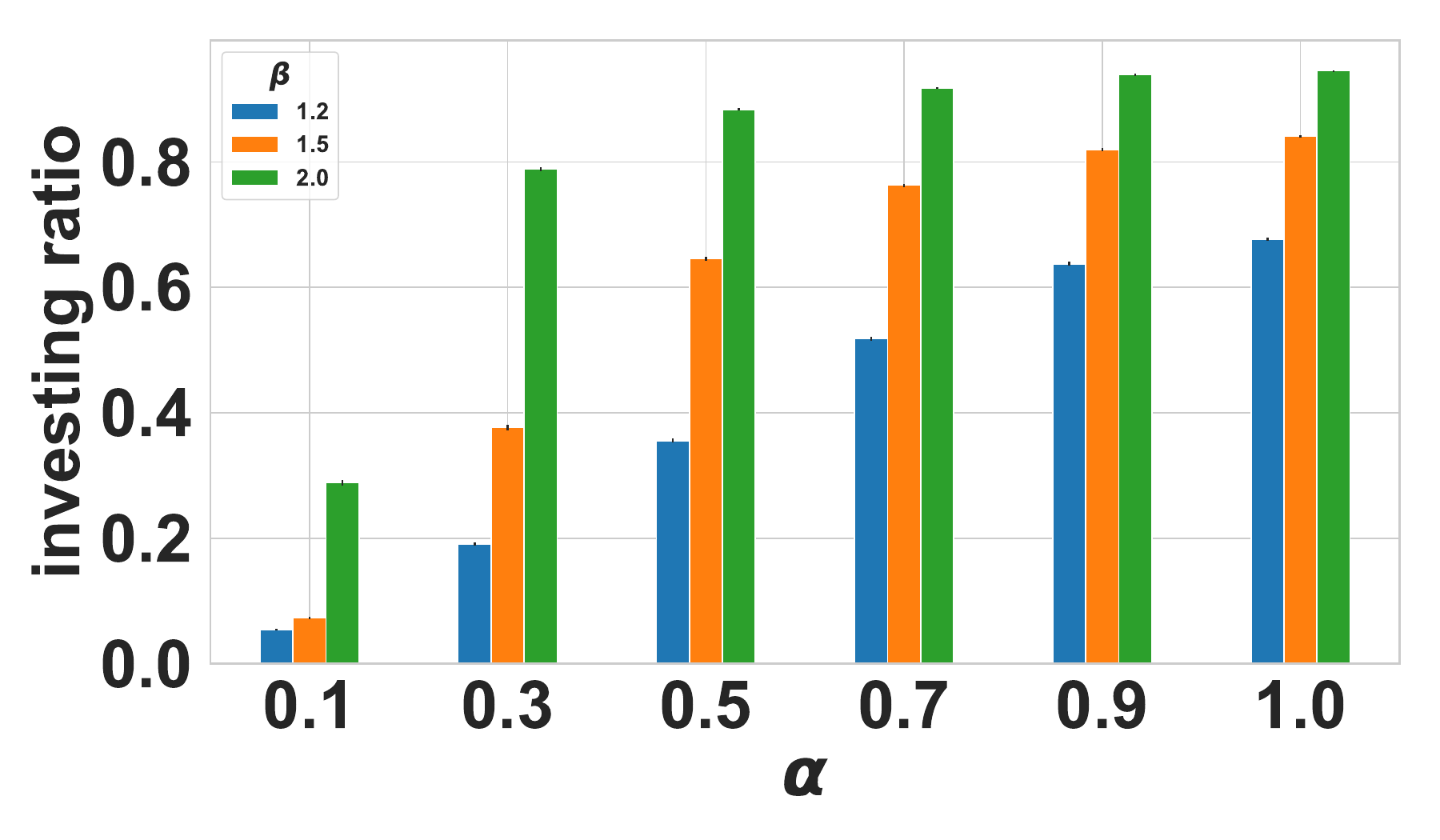} & \includegraphics[width=0.3\columnwidth]{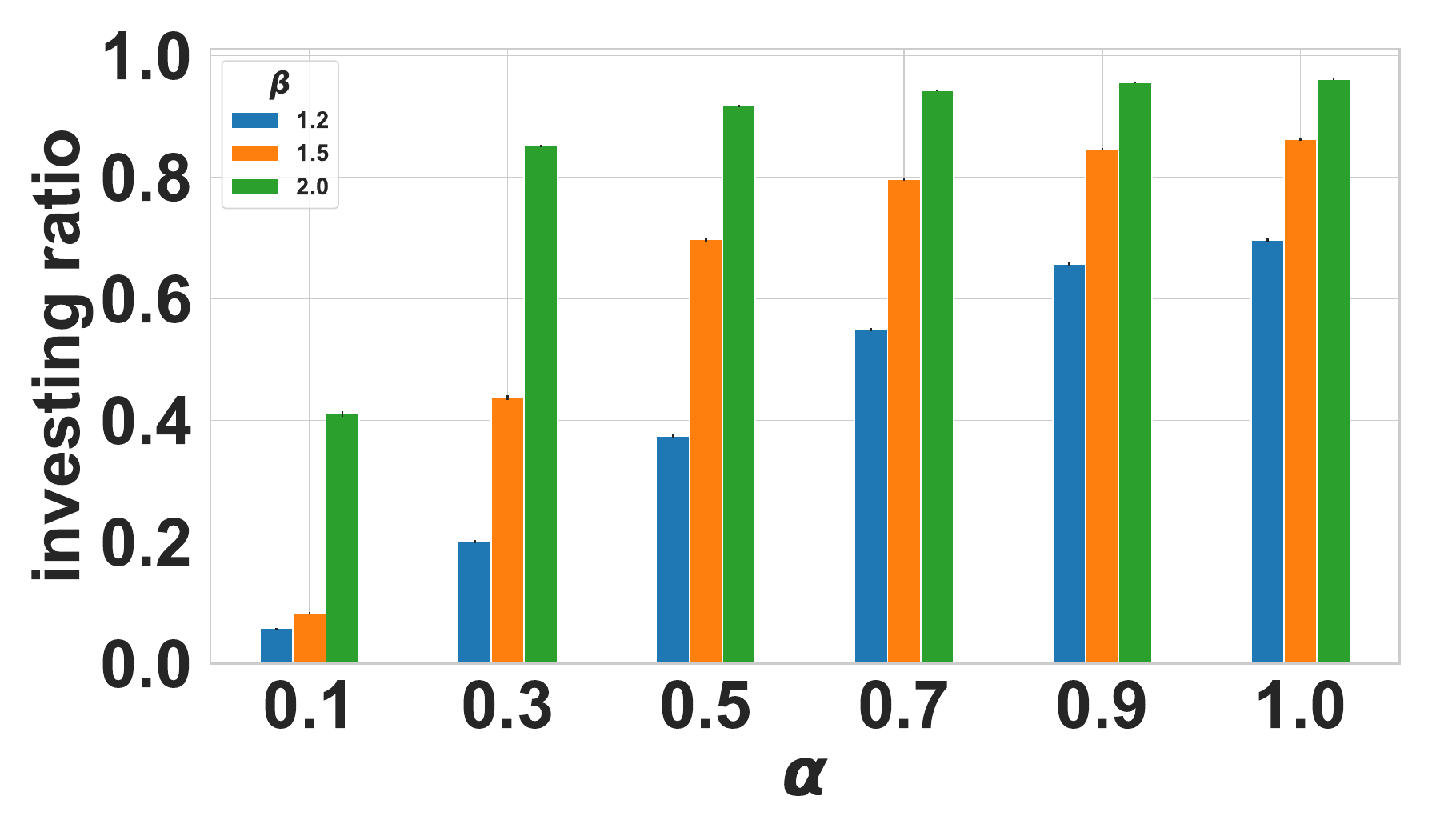}  & \includegraphics[width=0.3\columnwidth]{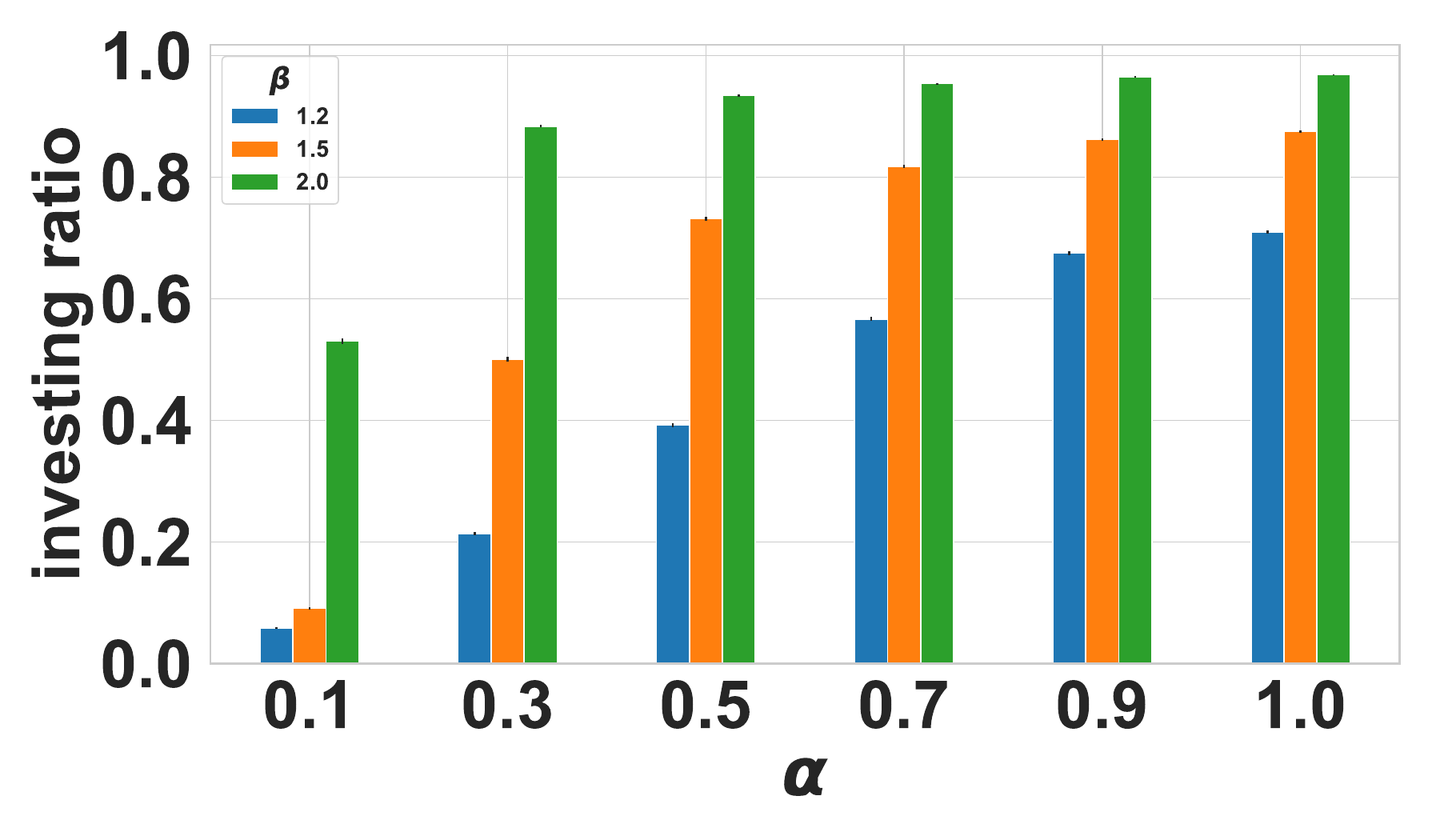} 
\end{tabular}

\small \caption{Ratio of investing players when $\gamma=0$ (top row) and $\gamma = 1$ (bottom row) on BA networks. From left to right: BA-1, BA-2, and BA-3.}
\label{fig:BA_investratio}
\end{figure}

\begin{figure}[h]
\centering
\begin{tabular}{lll}
\includegraphics[width=0.3\columnwidth]{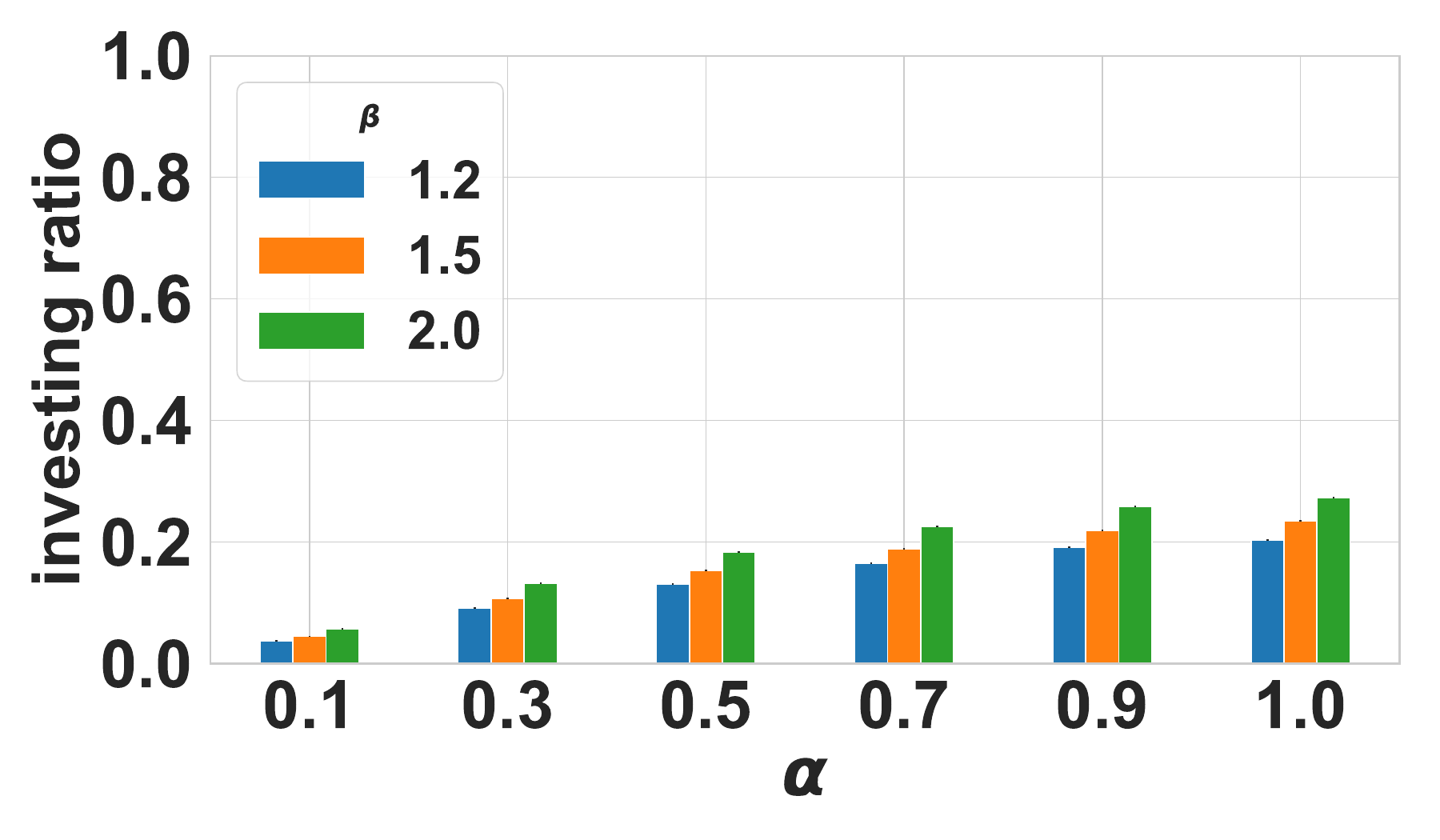} & \includegraphics[width=0.3\columnwidth]{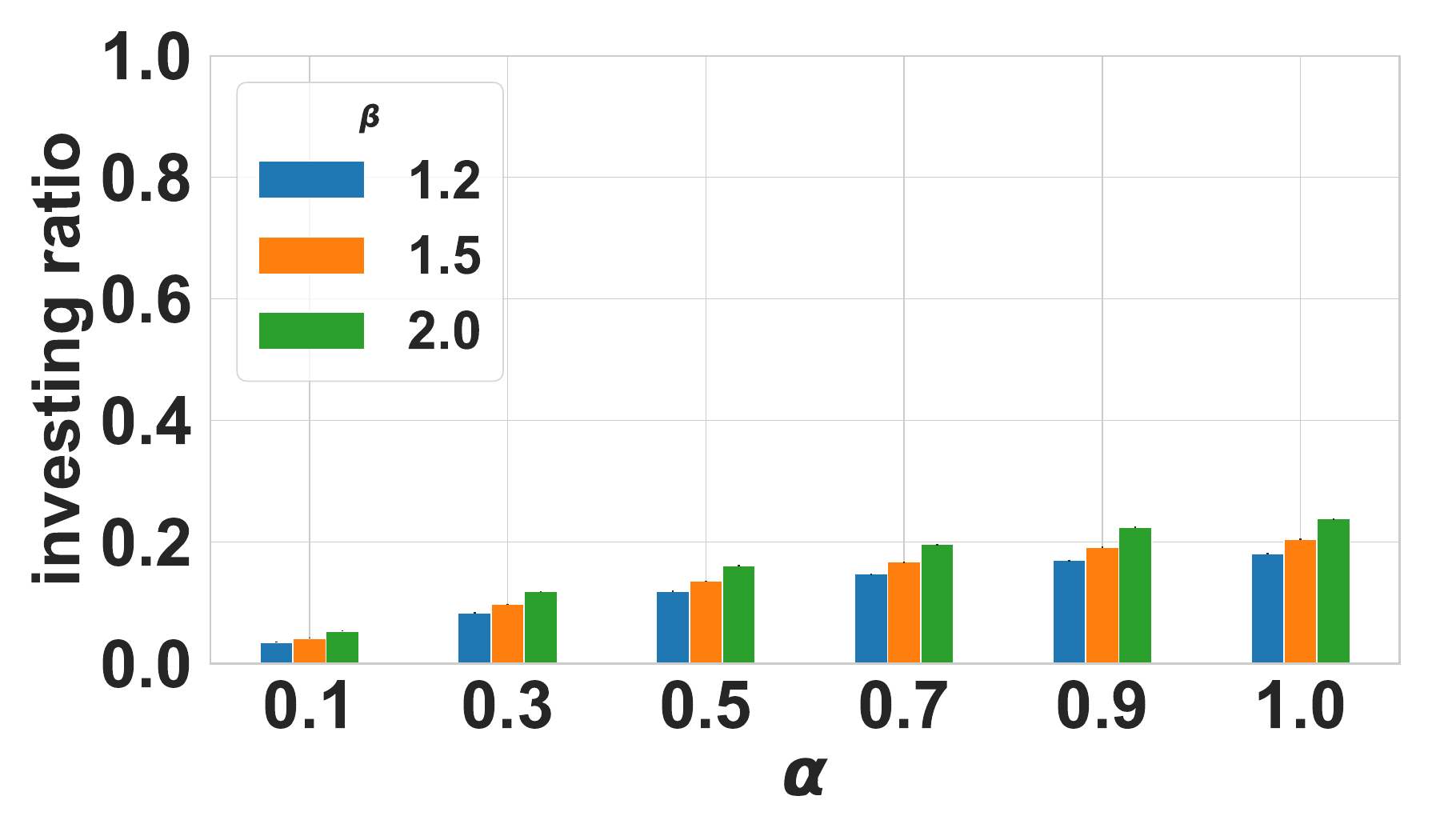} & \includegraphics[width=0.3\columnwidth]{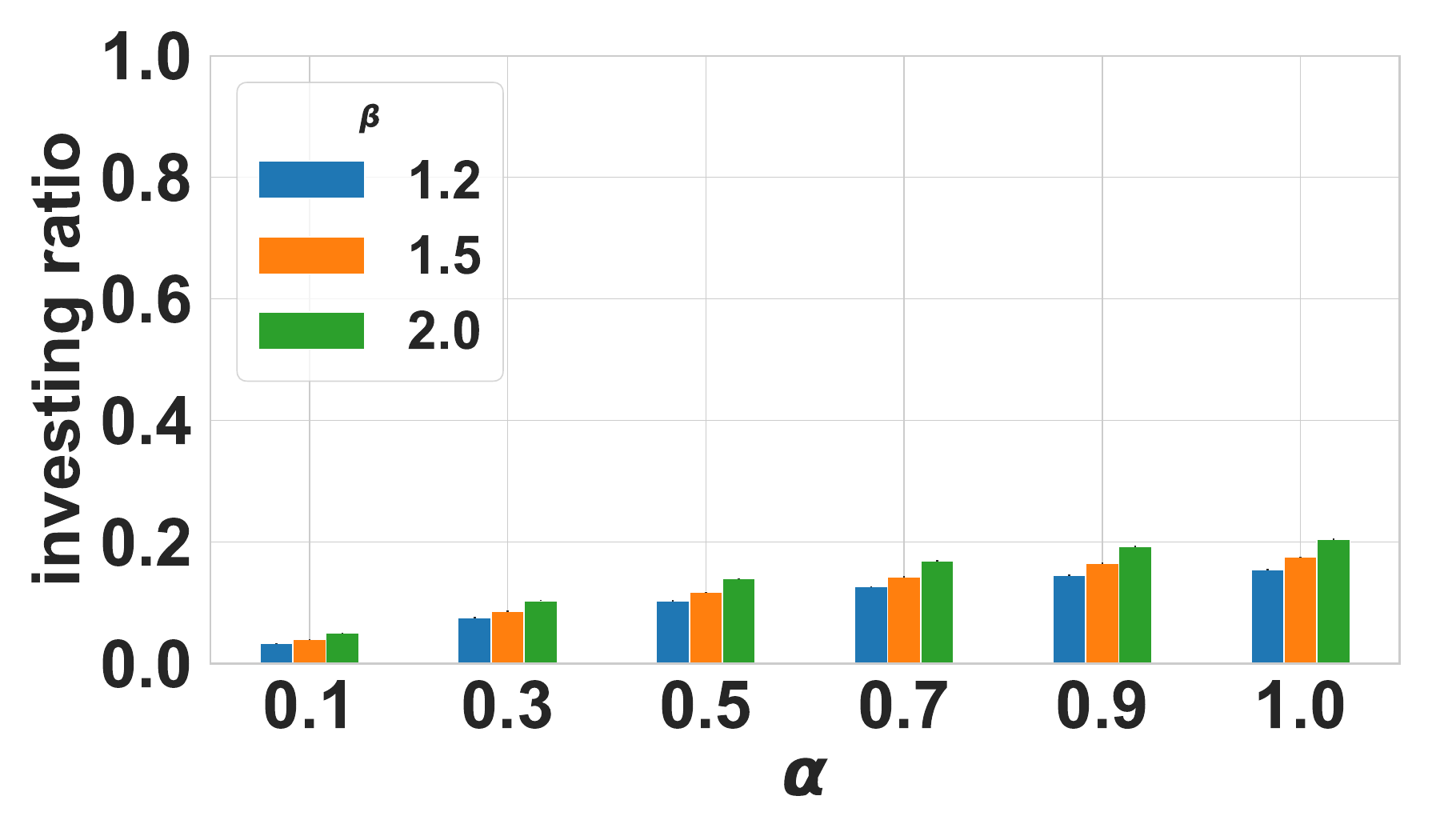} \\
\includegraphics[width=0.3\columnwidth]{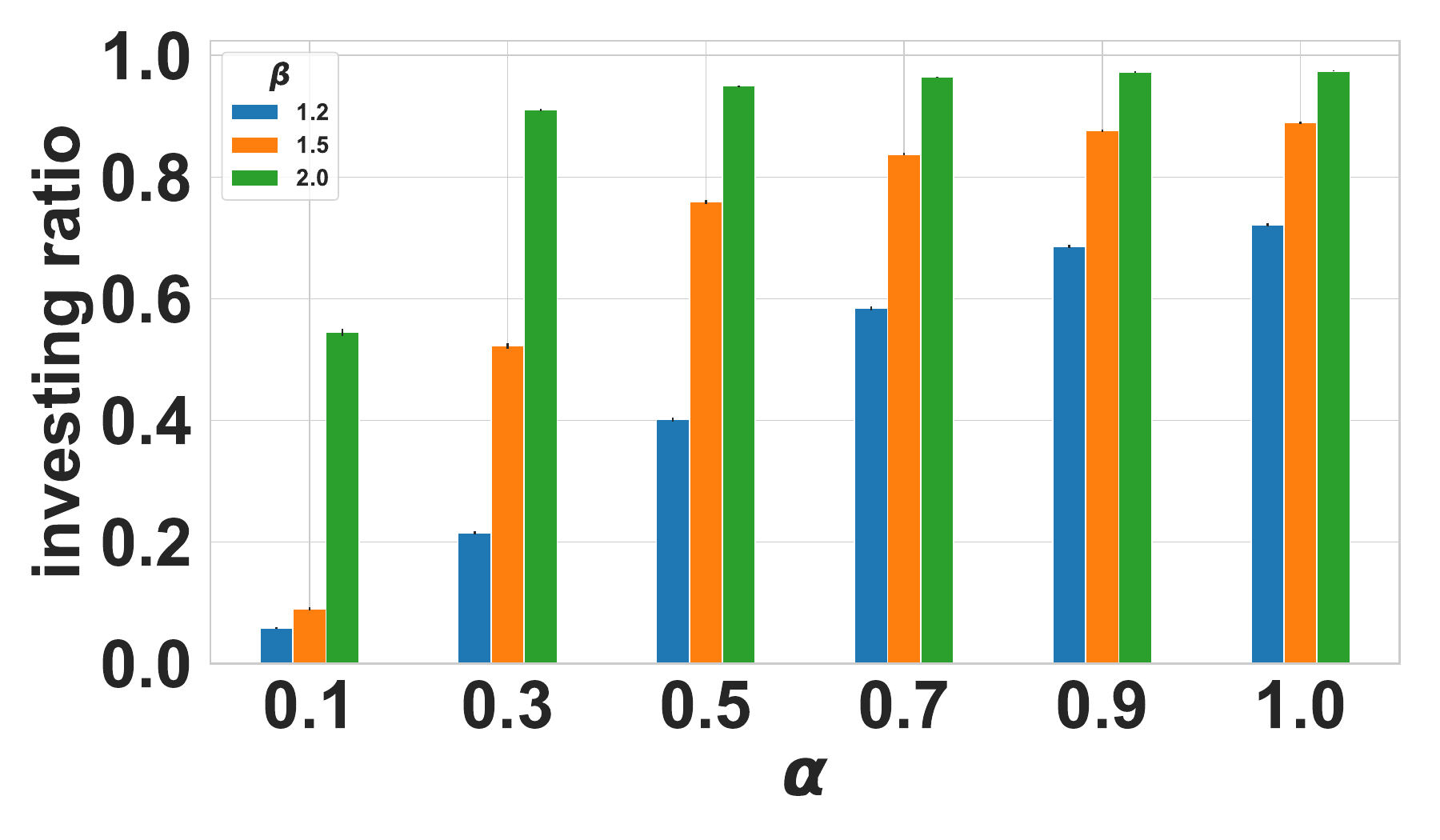} & \includegraphics[width=0.3\columnwidth]{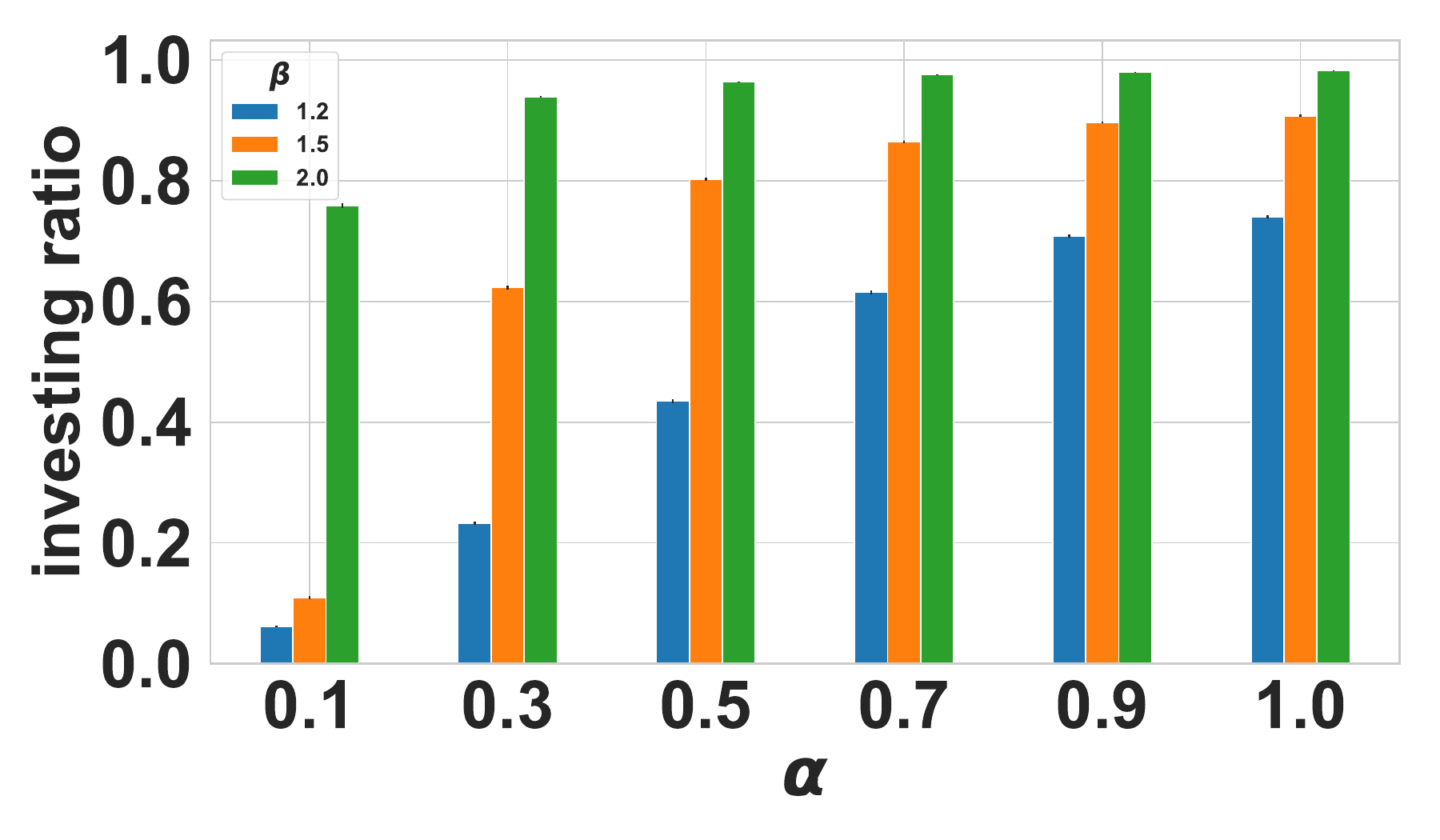}  & \includegraphics[width=0.3\columnwidth]{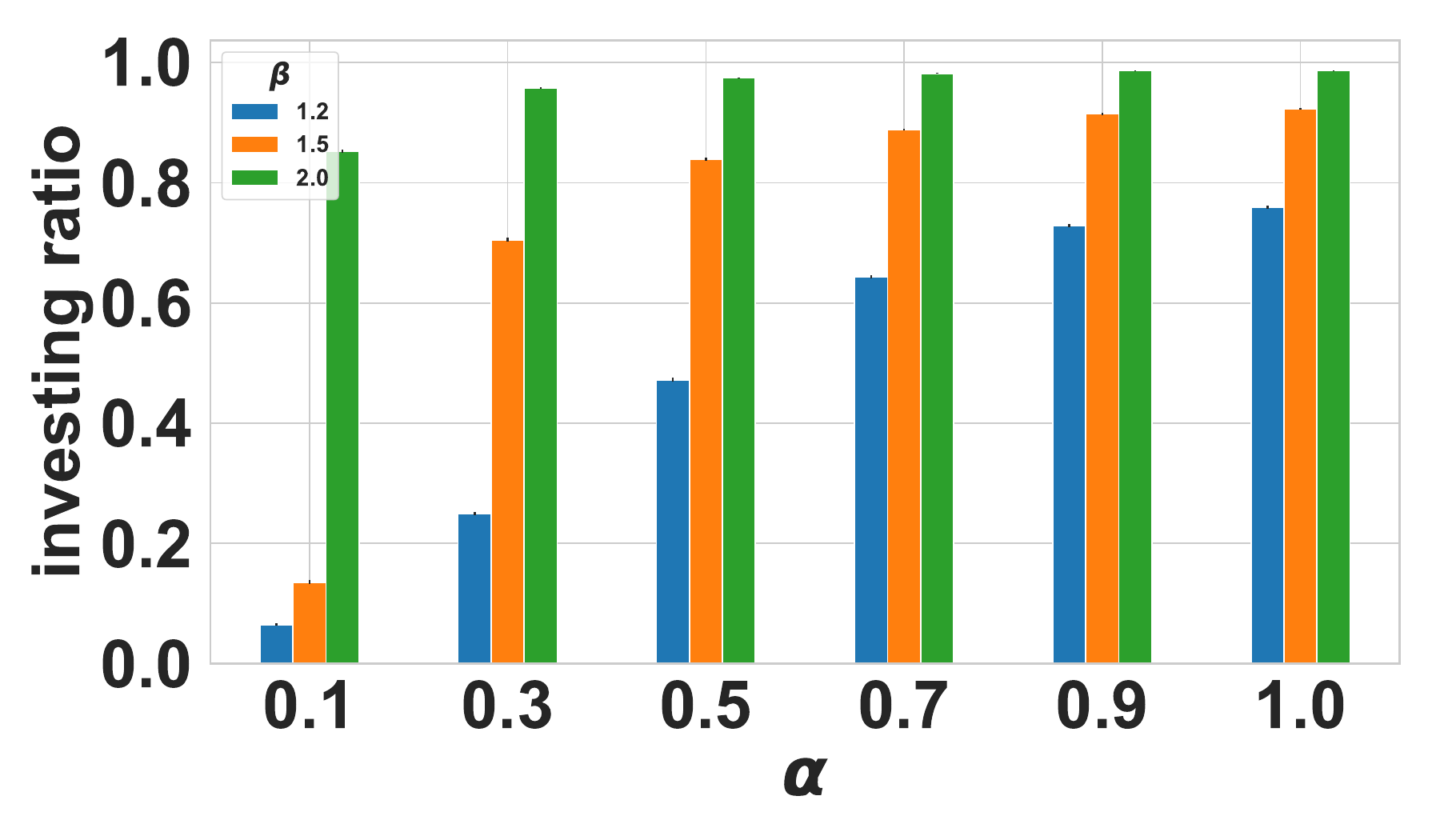} 
\end{tabular}

\small \caption{Ratio of investing players when $\gamma=0$ (top row) and $\gamma = 1$ (bottom row) on Small-World networks. From left to right: SW-1, SW-2, and SW-3.}
\label{fig:Small-World_investratio}
\end{figure}

\section{Conclusion}
We study binary networked public goods games from an algorithmic perspective.
We show that pure strategy equilibria may not exist even in highly restricted cases, and checking equilibrium existence is computationally hard in general.
However, we exhibit a number of important special cases in which such equilibria exist, can be efficiently computed, and even where socially optimal equilibria can also be efficiently found.
Finally, we presented a heuristic approach for approximately solving such games, shown to be highly effective in our experiments. 

\section*{Acknowledgments}
This work was partially supported by the National Science Foundation
(grant IIS-1903207) and Army Research Office (MURI grant W911NF1810208).

\balance
\bibliographystyle{abbrvnat}
\bibliography{paper.bib}

\end{document}